\documentclass[11pt]{article}

\usepackage{amsthm}
\usepackage{graphicx} % support the \includegraphics command and options
\usepackage{array} % for better arrays (eg matrices) in maths

\usepackage{amsmath, amssymb, amsfonts, verbatim}
\usepackage{hyphenat,epsfig,subfigure,multirow}
\usepackage{nicefrac}
\usepackage{paralist}

\usepackage[usenames,dvipsnames]{xcolor}
\usepackage[ruled]{algorithm2e}

\DeclareFontFamily{U}{mathx}{\hyphenchar\font45}
\DeclareFontShape{U}{mathx}{m}{n}{
      <5> <6> <7> <8> <9> <10>
      <10.95> <12> <14.4> <17.28> <20.74> <24.88>
      mathx10
      }{}
\DeclareSymbolFont{mathx}{U}{mathx}{m}{n}
\DeclareMathSymbol{\bigtimes}{1}{mathx}{"91}

\usepackage{tcolorbox}
\tcbuselibrary{skins,breakable}
\tcbset{enhanced jigsaw}

\usepackage[normalem]{ulem}
\usepackage[compact]{titlesec}

\definecolor{DarkRed}{rgb}{0.5,0.1,0.1}
\definecolor{DarkBlue}{rgb}{0.1,0.1,0.5}

\usepackage{nameref}
\definecolor{ForestGreen}{rgb}{0.1333,0.5451,0.1333}
\definecolor{Red}{rgb}{0.9,0,0}
\usepackage[linktocpage=true,
	pagebackref=true,colorlinks,
	linkcolor=DarkRed,citecolor=ForestGreen,
	bookmarks,bookmarksopen,bookmarksnumbered]
	{hyperref}
\usepackage[noabbrev,nameinlink]{cleveref}
\crefname{property}{property}{Property}
\creflabelformat{property}{(#1)#2#3}
\crefname{equation}{eq}{Eq}
\creflabelformat{equation}{(#1)#2#3}

\usepackage{bm}
\usepackage{url}
\usepackage{xspace}
\usepackage[mathscr]{euscript}

\usepackage{tikz}
\usetikzlibrary{arrows}
\usetikzlibrary{arrows.meta}
\usetikzlibrary{shapes}
\usetikzlibrary{backgrounds}
\usetikzlibrary{positioning}
\usetikzlibrary{decorations.markings}
\usetikzlibrary{patterns}
\usetikzlibrary{calc}
\usetikzlibrary{fit}
\usetikzlibrary{snakes}

\usepackage{mdframed}

\usepackage[noend]{algpseudocode}
\makeatletter
\def\BState{\State\hskip-\ALG@thistlm}
\makeatother

\usepackage{cite}
\usepackage{enumitem}

\usepackage[margin=1in]{geometry}

\renewcommand{\leq}{\leqslant}
\renewcommand{\geq}{\geqslant}

\newtheorem{theorem}{Theorem}
\newtheorem{lemma}{Lemma}[section]
\newtheorem{proposition}[lemma]{Proposition}

\newtheorem{claim}[lemma]{Claim}
\newtheorem{fact}[lemma]{Fact}

\newtheorem{definition}[lemma]{Definition}

\newtheorem*{claim*}{Claim}
\newtheorem*{proposition*}{Proposition}
\newtheorem*{lemma*}{Lemma}
\newtheorem*{problem*}{Problem}

\crefname{lemma}{Lemma}{Lemmas}
\crefname{claim}{Claim}{Claims}

\newenvironment{result}{\begin{mdframed}[backgroundcolor=lightgray!40,topline=false,rightline=false,leftline=false,bottomline=false,innertopmargin=2pt]
\begin{theorem}}{\end{theorem}\end{mdframed}}

\newtheorem{observation}[lemma]{Observation}

\newtheoremstyle{restate}{}{}{\itshape}{}{\bfseries}{~(restated).}{.5em}{\thmnote{#3}}
\theoremstyle{restate}

\allowdisplaybreaks

\renewcommand{\qed}{\nobreak \ifvmode \relax \else
      \ifdim\lastskip<1.5em \hskip-\lastskip
      \hskip1.5em plus0em minus0.5em \fi \nobreak
      \vrule height0.75em width0.5em depth0.25em\fi}

\newcommand{\eps}{\ensuremath{\varepsilon}}
\newcommand{\Paren}[1]{\Big(#1\Big)}

\newcommand{\bracket}[1]{\left[#1\right]}
\newcommand{\paren}[1]{\ensuremath{\left(#1\right)}\xspace}
\newcommand{\card}[1]{\left\vert{#1}\right\vert}

\newcommand{\IR}{\ensuremath{\mathbb{R}}}

\newcommand{\IN}{\ensuremath{\mathbb{N}}}
\newcommand{\IS}{\ensuremath{\mathbb{S}}}

\newcommand{\set}[1]{\ensuremath{\left\{ #1 \right\}}}
\newcommand{\poly}{\mbox{\rm poly}}

\DeclareMathOperator*{\Exp}{\ensuremath{{\mathbb{E}}}}
\DeclareMathOperator*{\Prob}{\ensuremath{\textnormal{Pr}}}
\renewcommand{\Pr}{\Prob}

\newenvironment{tbox}{\begin{tcolorbox}[
		enlarge top by=5pt,
		enlarge bottom by=5pt,
		 breakable,
		 boxsep=0pt,
                  left=4pt,
                  right=4pt,
                  top=10pt,
                  arc=0pt,
                  boxrule=1pt,toprule=1pt,
                  colback=white
                  ]%%
	}
{\end{tcolorbox}}

\newcommand{\event}{\ensuremath{\mathcal{E}}}

\newcommand{\II}{\ensuremath{\mathbb{I}}}

\newcommand{\mireal}[1][]{
  \ifx\relax#1\relax%
    \II(\mione \,; \mitwo)%
  \else%
    \II(\mione \,; \mitwo\mid #1)%
  \fi
}

\newcommand{\E}{\mathop{\mathbb{E}}}

\newcommand{\PP}[1]{{P}^{#1}}

\newcommand{\defn}[1]{\textbf{#1}}

\newcommand{\Set}{\ensuremath{\mathcal{S}}}
\newcommand{\DS}{\ensuremath{\mathbf{D}}}

\newcommand{\RankOp}{\textsc{Rank}}

\renewcommand{\IS}[1]{\ensuremath{\mathcal{I}(#1)}}

\newcommand{\dense}{\ensuremath{\textnormal{\texttt{density}}}\xspace}

\newcommand{\pw}{\ensuremath{{w}}\xspace}
\newcommand{\rw}{\ensuremath{{r}}\xspace}

\newcommand{\Win}{\ensuremath{\textnormal{\texttt{Win}}}\xspace}

\newcommand{\TT}{\ensuremath{\mathcal{T}}}

\newcommand{\DP}{\ensuremath{D\!P}}
\renewcommand{\PP}{\ensuremath{P}}

\newcommand{\size}[2]{d^{#2}(#1)}

\newif\ifsub
\subtrue %for submissions
\ifsub
\newcommand{\todo}[1]{}
\else
\newcommand{\todo}[1]{{\color{blue} #1 \marginpar{todo}}}
\fi
\title{Tight Bounds for Monotone Minimal Perfect Hashing}
\author{Sepehr Assadi\\Rutgers University \and Mart\'{\i}n Farach-Colton\\Rutgers University \and William Kuszmaul\\MIT}
\date{}

\begin{document}
\maketitle

\pagenumbering{roman}
% !TeX root = main.tex 
%!TEX root = main.tex

\begin{abstract}

\bigskip

The \textbf{monotone minimal perfect hash function} (MMPHF) problem is the following indexing problem. Given a set $S = \{s_1,\ldots,s_n\}$ of $n$ distinct keys from a universe $U$ of size $u$, create a data structure $\DS$ that answers the following query:

\[\RankOp(q) = \begin{cases}
  \text{rank of } q \text{ in } S  & q\in S\\
  \text{arbitrary answer} & \text{otherwise.}
\end{cases}
\]
Solutions to the MMPHF problem are in widespread use in both theory and practice.  

The best upper bound known for the problem encodes $\DS$ in $O(n\log\log\log u)$ bits and performs queries in $O(\log u)$ time.  It has been an open problem to either improve the space upper bound or to show that this somewhat odd looking bound is tight.

In this paper, we show the latter: any data structure (deterministic or randomized) for monotone minimal perfect hashing of any collection of $n$ elements from a universe of size $u$ requires $\Omega(n \cdot \log\log\log{u})$ expected bits to answer every query correctly. 

We achieve our lower bound by defining a graph $\mathbf{G}$ where the nodes are the possible ${u \choose n}$ inputs and where two nodes are adjacent if they cannot share the same $\DS$.  The size of $\DS$ is then lower bounded by the log of the chromatic number of $\mathbf{G}$.  Finally, we show that the fractional chromatic number (and hence the chromatic number) of $\mathbf{G}$ is lower bounded by $2^{\Omega(n \log\log\log u)}$.

\end{abstract}
\clearpage

\setcounter{tocdepth}{3}
\tableofcontents

\clearpage

\pagenumbering{arabic}
\setcounter{page}{1}

% !TeX root = main.tex 
%!TEX root = main.tex

\section{Introduction} 

The \textbf{monotone minimal perfect hash function} (MMPHF) problem is the following indexing problem. Given a set $S = \{s_1,\ldots,s_n\}$ of $n$ distinct keys from a universe $U$ of size $u$, create a data structure $\DS$ that answers the following query:

\[\RankOp(q) = \begin{cases}
  \text{rank of } q \text{ in } S  & q\in S\\
  \text{arbitrary answer} & \text{otherwise.}
\end{cases}
\]
The name of the problem comes from interpreting the data structure \DS{} as a hash function: given a sorted array 
$A = [a_1,\ldots,a_n]$, $\DS$ is a function mapping each $a_i$ to its position $i$. Such a hash function is \emph{minimal}, meaning that it maps $n$ items to $n$ distinct positions, and \emph{monotone}, meaning that whenever $a_i < a_j$ we have $\DS(a_i) < \DS(a_j)$, and vice versa. 

It may seem strange at first glance that $\DS$ is permitted to return arbitrary answers on negative queries. A key insight, however, is that this relaxation allows for asymptotic improvements in space efficiency: whereas the set $\Set$ would require $\Omega(n \log (u/n))$ bits to encode, Belazzougui, Boldi,  Pagh and Vigna~\cite{BelazzouguiBPV09} show that it is possible to construct an MMPHF $\DS$ using as few as $O(n \log \log \log u)$ bits, while supporting $O(\log u)$-time queries. 

The remarkable space efficiency of MMPHF makes it useful for a variety of practical applications (e.g., in  security \cite{boldyreva2011order}, key-value stores \cite{lim2011silt} and information retrieval \cite{navarro2014spaces}).  A high-performance implementation can be found in the Sux4J library \cite{sux4j, belazzougui2008theory}. MMPHF has also been widely used in the theory community for the design of space-efficient combinatorial pattern-matching algorithms (see, e.g., ~\cite{belazzougui2014alphabet, gagie2020fully, belazzougui2014linear, belazzougui2015optimal, clifford2015dictionary, belazzougui2020linear, belazzougui2016fully, grossi2010optimal}).

Despite the widespread use of MMPHF, it remains an open question \cite{BelazzouguiBPV09, Boldi15, dietzfelbinger20184} to determine the optimal bounds for solving this problem. 
%to whether the bound of $O(n \log \log \log u)$ bits is optimal. 
The best  lower bound achieved so far \cite{belazzougui2008theory, dietzfelbinger20184} is $\Omega(n)$ bits (which follows immediately from the same lower bound for minimal perfect hashing \cite{mehlhorn1982program}). Even disregarding applications (and the running time to answer queries), the information-theoretic question as to how many bits a MMPHF requires has been posed as a problem of independent combinatorial interest \cite{dietzfelbinger20184}. 

\paragraph{Our result.} We fully settle this question by establishing the following result: 
\begin{result}[Formalized in~\Cref{thm:main}]\label{res:main}
	Any data structure (deterministic or randomized) for monotone minimal perfect hashing of any collection of $n$ elements from a universe of size $u$ requires $\Omega(n \log\log\log{u})$ expected bits to answer every query correctly. The lower bound holds whenever $u$ is at least $n^{1+1 / \sqrt{\log n}}$ and at most $\exp\paren{\exp(\poly(n))}$. 
\end{result}

Thus, somewhat surprisingly, the $O(n \log \log \log u)$ bound achieved by \cite{belazzougui2008theory} is asymptotically optimal. We also note that the boundary conditions on $u$ in~\Cref{res:main} are natural in the following sense.  There are two trivial solutions for the MMPHF.  One encodes the entire input set $S$ in $O(u)$ bits and the other builds a perfect hash table mapping from elements of $S$ to their rank using $O(n\log{n})$ expected space.
Thus, when $u$ is very small, for example $u = O(n)$, the first solution achieves $O(u) = o(n\log\log\log u)$ bits.  On the other hand, when $u$ is very large, that is when $u$ is even beyond $\exp(\exp(\poly(n)))$, then the $O(n \log n)$-bit solution uses $o(n\log\log\log u)$ bits.  Our lower bound in~\Cref{res:main} covers almost the entire range in between.

The lower bound achieved by \Cref{res:main} is remarkably general: it applies independently of the running time of the data structure; and it applies even to randomized data structures that are permitted to store their random bits for free.

\paragraph{Our techniques.} The most intuitive approach
toward proving a lower bound of $d$ bits on the size 
of an MMPHF is to encode a $d$-bit string into the state of the data structure. This approach is already hindered by the fact that MMPHFs only support \emph{positive queries}, however. If the user already knows which elements are in the input, then the MMPHF encodes no interesting information --- but if the user only has partial information about the input, then the user can only get useful information from a small portion of possible MMPHF queries. The previous $\Omega(n)$ lower bound of~\cite{mehlhorn1982program, belazzougui2008theory, dietzfelbinger20184} addresses this as follows: consider any bit-string $x \in \set{0,1}^d$ and define: 
\[S(x) := \set{3,6,\ldots,3d} \cup \{3i + 1 \mid i \in [d], x_i = 1\} \cup \{3i - 1 \mid i \in [d], x_i = 0\}.\]
For every $i \in [d]$, firstly, $3i$ belongs to $S(x)$ and thus is a positive query, and secondly, $\RankOp(3i) = 2 \cdot (i-1) + x_i$. This allows us to recover $x$ from any MMPHF for $S(x)$, proving a lower bound of $d = \Omega(n)$ bits for MMPHF on size-$n$ subsets of universe $[3n+1]$. 
This approach, however, seems to be stuck at proving any $\omega(n)$ lower bound as these ``direct encodings'' ignore the delicate interaction between different elements in the input set\footnote{\emph{Any} lower bound of $d$ bits for a data structure immediately implies an encoding of $d$-bit strings in the state of the data structure by just assigning one bit-string to each state. This means that there is never a formal proof that one \emph{cannot} encode a bit-string in a data structure and still prove a lower bound.}. 

To get around these obstacles, we take a fundamentally different approach to proving \Cref{res:main}. We construct a ``conflict graph'' $G$ whose vertices are all the possible inputs to an MMPHF problem for a fixed $n$ and $u$. Two vertices are adjacent in $G$ if they cannot have the same MMPHF representation, that is, if the vertices share an element but with a different rank. Any MMPHF induces a proper coloring of this graph, where the color of a vertex corresponds to its MMPHF representation. As a result, the \emph{chromatic number} of the conflict graph is a lower bound on how many different MMPHF representations we must have, which implies that some input must have a representation of size at least $\log\chi(G)$ bits. This reduces our task to the combinatorial problem of lower bounding $\chi(G)$.\footnote{Slightly more care must be taken when bounding the \emph{expected} size of a MMPHF that is permitted to take different sizes on different inputs.}

The problem of bounding chromatic number of graphs defined over these types of set-systems has a rich history in the discrete math literature; see, e.g.~\cite{erdHos1966chromatic,furedi1992interval,duffus1995shift,simonyi2011directed}. For instance, Erd\H{o}s and Hajnal~\cite{erdHos1966chromatic} study \emph{shift-graphs} that have vertices corresponding to $n$-element subsets of $[u]$ and edges between vertices $(a_1,a_2,\ldots,a_n)$ and $(a_2,\ldots,a_n,a_{n+1})$ for all $a_1 < a_2 < \ldots < a_{n+1}$. They prove that 
the chromatic number of the shift-graph is $(1+o(1)) \cdot \log^{(n-1)}(u)$, namely, the $(n - 1)$-th iterated logarithm of $u$. The shift-graph is a subgraph of our conflict graph. Thus, by taking $u = 2 \upuparrows (n+1)$, i.e., the \emph{tower} of twos of height $n+1$, we can have $\chi(G) = 2^{\omega(n)}$, and thus prove an $\omega(n)$ lower bound for MMPHF on $n$-subsets of (extremely large) universes of size $u=2\upuparrows (n+1)$. This is the starting point of our approach. We now need to dramatically decrease the size of the universe, while also dramatically increasing the bound on the chromatic number by considering the conflict graph itself, and not only its shift-subgraph. 

To lower bound the chromatic number of the conflict graph, we consider the relaxation of this problem via \emph{fractional colorings} (see~\Cref{sec:fc}). Given that this latter problem can be formulated as a linear program (LP), a natural way for proving a lower bound on its value is to exhibit a feasible \emph{dual} solution instead\footnote{This is an inherently different technique than the one used in~\cite{erdHos1966chromatic} for the shift-graph, as it is known that the fractional chromatic number of the shift-graph is $O(1)$ (see, e.g.~\cite{simonyi2011directed}).}. This corresponds to the following problem: exhibit a distribution on vertices of the graph so that 
for any independent set, the probability that a vertex sampled from the distribution belongs to the independent set is bounded by $p$; this then implies that the fractional chromatic number (and in turn the chromatic number) are lower bounded by $1/p$. The main technical novelty of our work lies in the introduction of a highly non-trivial such distribution and the analysis of this probability bound for each independent set (we postpone the overview of this part to~\Cref{sec:overview} after we setup the required background). This allows us to lower bound the (fractional) chromatic number of the conflict-graph by $\Omega(n\log{n})$ when the universe is of size $u=2^{2^{\text{poly}(n)}}$ which gives an $\Omega(n\log\log\log{u})$ lower bound for MMPHF on such universes.  

Working with fractional colorings, beside being an immensely helpful analytical tool, has several additional benefits for us. Firstly, unlike standard (integral) colorings, fractional colorings admit a natural \emph{direct product} property over a certain union of graphs; this allows us to extend the lower bound for MMPHF from universes of size doubly exponential in $n$ (which are admittedly not the most interesting setting of parameters), all the way down to universes of size $n^{1+o(1)}$. Secondly, unlike the (integral) chromatic number, which yields a lower bound only on the space of deterministic MMPHFs, we show that lower bounding the fractional chromatic number allows us to prove a lower bound even for randomized MMPHFs that have access to their randomness for free. We believe this technique, namely, defining a proper conflict graph and bounding its fractional coloring by exhibiting a feasible dual solution, may be applicable to many other data structure problems and is therefore interesting in its own right.

% !TeX root = main.tex 
%!TEX root = main.tex

\section{Preliminaries}\label{sec:prelim}

\paragraph{Notation.} For any integer $t\geq s \geq 1$, we let $[t]:=\set{1,\ldots,t}$ and let $[s,t] = \{s, \ldots, t\}$. For a tuple $(X_1,\ldots,X_t)$, we further define  
$X_{<i} := (X_1,\ldots,X_{i-1})$ and $X_{-i} := (X_1,\ldots,X_{i-1},X_{i+1},\ldots,X_t)$.

\subsection{Problem Definition and Model of Computation}\label{sec:problem} 
For any integer $n,u \geq 1$, we let $\DS(n,u)$ be an MMPHF indexing algorithm for size-$n$ subsets of $[u]$.  That is, if $\mathcal{S}_{n,u} = \{ S \subseteq [u] ~\text{s.t.}~ |S| = n\}$ then  for all $S \in \mathcal{S}_{m,u}$, $\DS(S)$ is the MMPHF index for $S$.  

For any fixed choice of random bits $r$, we use $\DS^r$ to denote the resulting MMPHF with random bits $r$. Note that for any fixed choice of $r$, $\DS^r$ is deterministic. 
For any $S \in \mathcal{S}_{n,u}$ and randomness $r$, define $\size{S}{r}$ as the size in bits of the MMPHF index  $\DS^r(S)$. Define:
	\[
	    d (n,u):= \max_{S \in \mathcal{S}_{n,u}} \Exp_r\bracket{\size{S}{r}}.
	\]  
	When $n$ and $u$ are clear, we drop them and refer simply to $\DS$ and $d$.
	
    In this definition of size, we are giving the MMPHF a big advantage: we are not charging the algorithm for storing its randomness. In other words, the algorithm has access to a tape of random bits chosen independent of the input that it can use for both 
    creating the index as well as answering the queries. Furthermore, we also allow the algorithm unbounded computation time. Thus, the only measure of interest
    for us is the \emph{size} of the index. Finally, any deterministic MMPHF in this model is simply a randomized MMPHF that ignores its random bits and thus we will only focus on randomized MMPHFs from now on. 

\subsection{Fractional Colorings}\label{sec:fc}
A key tool that we use in establishing our lower bound is the notion of a \textbf{fractional coloring} of a graph. We now review the basics of fractional colorings, which we need in our proofs. 

Let $G=(V,E)$ be any undirected graph. A proper coloring of $G$ is any assignment of colors to vertices of $G$ so that no edge is monochromatic. The chromatic number $\chi(G)$ is the minimum 
number of colors in any proper coloring of $G$. The fractional relaxation of chromatic number can then be defined as follows.

Let $\IS{G} \subseteq 2^{V}$ denote the set of all independent sets in $G$, and for any vertex $v \in V$, define $\IS{G,v}$ as the set of all independent sets that contain the vertex $v$. 
 A fractional coloring of $G$ is any assignment of $x = (x_1, 
 \ldots, x_{\IS{G}}), 0\leq x_i \leq 1$, to the independent sets of $G$ satisfying the following constraint: 
\begin{align*}
	\text{for every vertex $v \in V$:} \quad \sum_{I \in \IS{G,v}} x_I \geq 1. %\label{eq:fractional covering}
\end{align*}
The \defn{value} $\card{x}$ of a fractional coloring $x$ is given by $\sum_{I \in \IS{G,v}} x_I$. 

The \defn{fractional chromatic number} $\chi_f(G)$ is the minimum value of any fractional coloring of $G$. This quantity can be formalized as a linear program (LP): 
%\begin{subequations}\label{eq:chif}
\begin{align}\label{eq:chif}
	\chi_f(G) := \min_{x \in \IR^{\IS{G}}_{\geq 0}} &\sum_{I \in \IS{G}}\hspace{-5pt} x_I \quad \text{subject to} \quad \sum_{v \in \IS{G,v}}\hspace{-8pt} x_I \geq 1 \quad \forall v \in V. % x_I \geq 0 &\hspace{-50pt} \forall I \in \IS{G}\\
%	&\hspace{-5pt}\sum_{v \in \IS{G,v}}\hspace{-8pt} x_I \geq 1 &\hspace{-50pt} \forall v \in V \\
%	&\hspace{5pt} x_I \geq 0 &\hspace{-50pt} \forall I \in \IS{G}. 
\end{align}
%\end{subequations}

Any proper coloring of $G$ with $k$ colors induces a solution $x$ of value $k$ to this LP, where $x_I$ is set to $1$ for the independent sets $I$ that correspond to color classes in the coloring.  Thus the LP given by \Cref{eq:chif} is indeed a relaxation of the original coloring problem. %, leading to the following standard fact. 
\begin{fact}\label{fact:chif-chi}
	For any graph $G$, $\chi_f(G) \leq \chi(G)$. 
\end{fact}
It is worth mentioning that at the same time $\chi(G) = O(\log{\card{V(G)}}) \cdot \chi_f(G)$ using the standard randomized rounding argument (we do not use this direction explicitly in our paper). 

A primal-dual analysis of the fractional-chromatic-number LP implies the following results. 
These results are standard but we provide proofs in \Cref{app:fc} for completeness. 

\begin{proposition}\label{prop:chif-product}
	Let $G_1 = (V_1,E_1)$ and $G_2 = (V_2,E_2)$ be arbitrary graphs. Define $G_1 \vee G_2$ as a graph on vertices $V_1 \times V_2$ and define an edge between vertices $(v_1,v_2)$ and $(w_1,w_2)$ whenever $(v_1,w_1)$ is an edge in $G_1$ 
	\underline{or} $(v_2,w_2)$ is an edge in $G_2$. Then, $\chi_f(G_1 \vee G_2) = \chi_f(G_1) \cdot \chi_f(G_2)$. 
\end{proposition}
\Cref{prop:chif-product} allows us to determine $\chi_f$ of a product of several graphs by focusing on each individual graph separately.

\begin{proposition}\label{prop:chif-distribution}
	For any graph $G=(V,E)$, 
	\[
	    \chi_f(G) = \max_{\textnormal{distribution $\mu$ on $V$}} ~ \min_{I \in \IS{G}} ~ \Paren{\Pr_{v \sim \mu} \paren{v \in I}}^{-1}.
	\]
\end{proposition}
\Cref{prop:chif-distribution} provides us with a tool to lower bound $\chi_f$ by finding a suitable distribution on the vertices so that no independent set has a significant probability of being by the distribution. 

%%\subsection{Randomized MMPHF}
% !TeX root = main.tex 
%!TEX root = main.tex

\section{A Lower Bound for MMPHF via Fractional Colorings}\label{sec:lower-setup}

We can now formally state the main theorem of this paper. 

\begin{theorem}[Formalization of~\Cref{res:main}]\label{thm:main}
	For any $n,u \in \IN^+$ such that $n \cdot 2^{\sqrt{\log{n}}} \leq u \leq 2^{n^{n^2+n}}$, 
	and for any MMPHF algorithm $\DS(n,u)$, 
	\[
	d(n,u) = \Omega(n\log\log\log u).
	\]
\end{theorem}

The rest of the paper presents the proof of Theorem \ref{thm:main}. We spend the rest of the section reframing the theorem in terms of the fractional chromatic number of a certain graph associated with MMPHF problem---we will then show how to lower bound the fractional chromatic number in the next section.

\subsection{Conflict Graph and its Fractional Chromatic Number}\label{sec:conflict-graph} 

Let $m \geq 1$ be an integer and define $M:= 2^{m^{m^2+m}}$. Define the graph $G(m) := (V(m),E(m))$ as: 
\begin{itemize}
	\item The vertex set is $V(m) = \mathcal{S}_{m,M}$, that is, the size-$m$ subsets of  $[M]$. We denote each vertex $v \in V(M)$ by the $m$-tuple $v:= (v_1,\ldots,v_m)$ where $0 < v_1 < v_2 < \cdots < v_m \le M$. 
	\item The edge set $E(m)$ is defined as follows. Let $v=(v_1,\ldots,v_m)$ and $w=(w_1,\ldots,w_m)$ be any two vertices in $V(M)$. Then, there is an edge $(v,w) \in G(m)$ iff 
	there exists some pair of indexes $i \neq j \in [m]$ such that $v_i = w_j$. 
\end{itemize}
We refer to $G(m)$ as the \textbf{conflict graph} of $m$.  
%Let $\DS(v)$ be the MMPHF for input $v\in V(m)$. 
The following lemma clarifies our interest in this graph by showing that
fractional chromatic number of $G(m)$ can be used to lower bound size of any MMPHF (for certain parameters of input). 

\begin{lemma}\label{lem:conflict-graph}
	Let $m \geq 1$ be an integer and let $M = 2^{m^{m^2+m}}$. Consider any MMPHF $\DS(m,M)$.  Then 
	\[
     d(m,M) \geq (\log{\chi_f(G(m))}-2)/2.
     \]
\end{lemma}
\begin{proof}
	Consider any two vertices $ v,w \in G(m)$. If there is an edge between $v$ and $w$, then there exists an element $z = v_i = w_j, i\not= j$.  Therefore for every choice of randomness $r$, $\DS^r(v) \not= \DS^r(w)$, because query $z$ must return $i$ on $\DS^r(v)$ and $j$ on $\DS^r(w)$. This implies that
	for every $r$, the set of vertices $v$ with the same $\DS^r(v)$ form an independent set in $G(m)$ (and the collection of these sets
	is a coloring of $G(m)$). We use $\mathcal{I}^r$ to denote these independent sets in $G(m)$ for this choice of $r$. 

	On the other hand, by~\Cref{prop:chif-distribution}, there exists a distribution $\mu$ on $V(m)$ such that
	\begin{align}
	    \chi_f(G(m)) = \min_{I \in \IS{G(m)}} \quad \Paren{\Pr_{v \sim \mu}\paren{v \in I}}^{-1}.   \label{eq:vm-chif}
	\end{align}
	Let us fix that distribution. Under this distribution, by the definition of $d$,
	\begin{align*}
	  d = d(m,M) = \max_{v \in V(m)} \Exp_r[d^r(v)] \geq \Exp_{v \sim \mu} \Exp_{r} \bracket{\size{v}{r}} = \Exp_{r} \Exp_{v \sim \mu} \bracket{\size{v}{r}}. 
	\end{align*}
    An averaging argument now implies that there exists a choice $r^*$ of random bits 
    such that 
    \[
        \Exp_{v \sim \mu}\bracket{\size{v}{r^*}} \leq d.
    \]
	By Markov's inequality, with probability at least $1/2$, for $v \sim \mu$, we have that $\size{v}{r^*} \leq 2d$. 
	
	Recall that $\DS^{r^*}(v)$ corresponds to an independent set in $\mathcal{I}_{r^*}$. Moreover, there can be at most $2^{2d+1}-2$ independent sets $I$ in $\mathcal{I}_{r^*}$ such that for all $v \in I$, $\size{v}{r^*} \leq 2d$; this is because there are at most $2^{2d+1}-2$ choices for $\DS^{r^*}(v)$ across
	all $v \in V(m)$ that can use up to $2d$ bits in their index (as the number of non-empty binary strings of length at most $2d$ is $2^{2d+1}-2$). 
	Since a random $v \sim \mu$ belongs to one of these $2^{2d+1}-2$ independent sets with probability at least half, we necessarily have some 
	independent set $I \in \mathcal{I}_{r^*}$ where
	\[
	    \Pr_{v \sim \mu}\paren{v \in I} \geq \frac{1}{2 \cdot (2^{2d+1}-2)} \geq \frac{1}{2^{2d+2}}. 
	\]
	Plugging in this bound in~\Cref{eq:vm-chif}, we have, 
	\[
	    \chi_f(G(m)) \leq 2^{2d+2},
	\]
	which implies that $d \geq (\log{\chi_f(G(m)}-2)/2$, concluding the proof.  
\end{proof}

\Cref{lem:conflict-graph} reduces our task of proving~\Cref{thm:main} to establishing a lower bound on $\chi(G(m))$. This will be accomplished by the following lemma, which we prove in \Cref{sec:distribution}.

\begin{lemma}\label{lem:cg-chif}
	There is an absolute constant $\eta > 0$ such that for every sufficiently large $m \geq 1$, 
	\[
		\chi_f(G(m)) \geq m^{\eta \cdot m}. 
	\]
\end{lemma}
By plugging in the lower bound of $\chi_f(G(m))$ from~\Cref{lem:cg-chif} inside~\Cref{lem:conflict-graph}, 
we get that for any sufficiently large $n \geq 1$ and universe size $u=2^{m^{m^2+m}}$, 
the lower bound on the MMPHF problem is $\Omega(n\log{n}) = \Omega(n\log\log\log{u})$ as $\log{n} = \Theta(\log\log\log{u})$ here. 

Thus Lemmas \ref{lem:cg-chif} and \ref{lem:conflict-graph} can be combined to prove~\Cref{thm:main} modulo a serious caveat: the lower bound  only holds for instances of the problem wherein the universe size 
is larger than doubly exponential in $n$, which is admittedly not the most interesting setting of the parameters. 
In the next subsection, we use a simple graph product argument (plus \Cref{prop:chif-product}) to extend this lower bound to the whole range of parameters $u$ considered by~\Cref{thm:main}. 

\subsection{Extending the MMPHF Lower Bound to Small Universes}\label{sec:small-universe}

For every pair of integers $m, \ell \geq 1$, define $G(m,\ell) = (V(m,\ell), E(m,\ell))$ as the \textbf{$\ell$-offset conflict graph} where the vertex set $V(m,\ell)$ is the set of all size-$m$ subsets of $[\ell+1, M+\ell]$, and the edge set $E(m,\ell)$ is defined as in normal conflict graphs. (Thus $G(m,0) = G(m)$.)

Furthermore, for every integer $m,k \geq 1$, we define the \textbf{$k$-fold conflict graph}, denoted by $G^{\oplus k}(m)$, as the graph: 
\[
	G^{\oplus k}(m) = (V^{\oplus k}(m), E^{\oplus k}(m)) := G(m,0) \vee G(m,M) \vee  G(m,2M) \vee \cdots \vee G(m,(k-1)M),
\]
where `$\vee$' denotes the graph product in~\Cref{prop:chif-product}.  The direct interpretation of the nodes of $V^{\oplus k}(m)$ is a product of tuples from disjoint ranges, but we can also interpret it as a single tuple of length $k \cdot m$. This way, $G^{\oplus k}(m)$ is a subset of the conflict graph on $km$-size subsets of $[k \cdot M]$ and
it makes sense to compute $\DS(v)$ for any $v\in V^{\oplus k}(m)$.

Therefore, by~\Cref{lem:conflict-graph}, we again have a lower bound of $\Omega(\log{\chi_f(G^{\oplus k}(m))})$ for MMPHF on tuples of length $n=km$ from a universe of size $u = kM$. 

By~\Cref{prop:chif-product}, combined with~\Cref{lem:cg-chif}, we have, 
\begin{align*}
 \log{\chi_f(G^{\oplus k}(m))} = k \cdot \log{\chi_f(G(m))} \geq \Omega(k\cdot m \cdot \log{m}) = \Omega(n\log{m}).
\end{align*}

Consider a choice of
\[
m=({\log\log{n}})^{1/6} ~~\text{and}~~ k=n/(\log\log{n})^{1/6},
\]
which in turn gives us 
\[
    u = k \cdot 2^{m^{m^2+m}} \ll k \cdot 2^{2^{m^3}} = \frac{n}{(\log\log{n})^{1/6}} \cdot 2^{2^{{\sqrt{\log\log{n}}}}} \ll n \cdot 2^{\sqrt{\log{n}}}.
\]
By the above equation, 
we have a lower bound of $\Omega(n\log\log\log{u})$ for MMPHF
given that in this case, $\log{m} = \Theta(\log\log\log{u})$. Thus, so far, we have proven~\Cref{thm:main} on both its boundary cases, namely, when 
$u = n \cdot 2^{\sqrt{\log{n}}}$ and when $u = 2^{n^{n^2+n}}$. The proof
can now be extended to the full range of the parameters in the middle by re-parameterizing $k$ appropriately; see \Cref{sec:allu} for the complete argument.

We conclude that in order to finish the proof of~\Cref{thm:main}, we need only establish~\Cref{lem:cg-chif}.

% !TeX root = main.tex 
%!TEX root = main.tex

\section{Fractional Chromatic Number of Conflict Graphs}\label{sec:distribution}

In this section, we establish a lower bound on the fractional chromatic number of the conflict graph $G(m)$ for any (large enough) $m \geq 1$, and we thereby prove~\Cref{lem:cg-chif}. 

\Cref{prop:chif-distribution} gives us a clear path for proving the lower bound on $\chi_f(G(m))$ given by~\Cref{lem:cg-chif}: we can design a distribution $\mu$ on vertices of $V(m)$ and then, for every independent set $I \in \IS{G(m)}$, we can upper bound the probability that $v$ sampled from $\mu$ belongs to $I$. As $\chi_f$ in~\Cref{prop:chif-distribution} is maximum over all possible distributions, our distribution provides a lower bound for $\chi_f(G(m))$. 

To continue, we need the following  interpretation of the (maximal) independent sets in $G(m)$. 

\begin{observation}\label{obs:mis-array}
	Any maximal independent set $I$ in $G(m)$ can be uniquely identified by a function $f_I : [M] \rightarrow [m]$ such that for every vertex $v=(v_1,\ldots,v_m) \in I$, $f_I(v_i) = i$. 
\end{observation}
\begin{proof}
	Consider any two vertices $v,w \in I$. Since there is no edge between $v=(v_1,\ldots,v_m)$ and $w=(w_1,\ldots,w_m)$ in $G(m)$, whenever $v_i = w_j$, we necessarily have that $i = j$. Thus, any element of $e \in [M]$ can only 
	appear in a single index $i_e \in [m]$ throughout all vertices $v \in I$ (or does not appear at all in $v$). We can thus define $f_I(e)$ to be $i_e$, giving us a functino $f_I$ with the desired property. 
	
	We now show that $f_I$ uniquely identifies $I$. If we define $I'$ to be the set of vertices $v = (v_1,\ldots,v_m) \in I$ satisfying $f_I(v_i) = i$ for all $i$, then $I'$ is an independent set satisfying $I \subseteq I'$. Since $I$ is assumed to be maximal, it follows that $I = I'$, meaning that we can recover $I$ from $f_I$.  
\end{proof}

\Cref{obs:mis-array} allows us to reduce \Cref{lem:cg-chif} to the following lemma about $m$-tuples of increasing integers. Proving Lemma \ref{lem:main} is the main technical contribution of our work. 

\begin{lemma}\label{lem:main}
	There is an absolute constant $\eta > 0$ such that for any sufficiently large $m \geq 1$ and $M=2^{m^{m^{2}+m}}$, the following is true. 
	There exists a distribution on $m$-tuples of increasing numbers $X_1 < \cdots < X_m$ from $[M]$ such that for any function $f: [M] \rightarrow [m]$, 
	\[
		\Pr_{(X_1,\ldots,X_m)}\paren{\text{$\forall i \in [m]:$  $f(X_i) = i$}} \leq m^{-\eta \cdot m}. 
	\]
\end{lemma}

Before proving~\Cref{lem:main}, we show how it implies~\Cref{lem:cg-chif}. 

\begin{proof}[Proof of~\Cref{lem:cg-chif} (assuming~\Cref{lem:main})]
	Any choice of $(X_1,\ldots,X_m)$ in~\Cref{lem:main} can be mapped to a unique vertex $v \in G(m)$ and vice versa. Thus, $(X_1,\ldots,X_m)$ induces a distribution $\mu$ on vertices $V(m)$: sample $(X_1,\ldots,X_m)$ and return the vertex $v=(v_1,\ldots,v_m)$ where $v_i = X_i$ for all $i \in [m]$. Moreover, for any maximal independent set $I \in \IS{G}$, by~\Cref{obs:mis-array}, the vertex corresponding to $(X_1,\ldots,X_m)$ belongs to $I$ iff $f_I(X_i) = i$ for all $i \in [m]$. Thus, 
	\[
		\Pr_{v \sim \mu}\paren{v \in I} = \Pr_{(X_1,\ldots,X_m)}\paren{\text{$\forall i \in [m]:$  $f(X_i) = i$}} \leq m^{-\eta \cdot m}. 
	\]
	As every independent set of $G(m)$ is a subset of some maximal independent set, 
	the upper bound continues to hold for every independent set in $G(m)$. 
	
	By~\Cref{prop:chif-distribution}, 
	\[
	    \chi_f(G(m)) \geq \min_{I \in \IS{G(m)}} \Paren{\Pr\paren{v \in I}}^{-1} \geq m^{\eta \cdot m},
	\]
concluding the proof. 
\end{proof}

The rest of the section proves~\Cref{lem:main}. We start with a high-level overview in \Cref{sec:overview}. We then define the distribution that we will use for the proof of~\Cref{lem:main} (Section \ref{sec:distribution}) and analyze it to establish \Cref{lem:main} (Section \ref{sec:analysis}). The probability distribution that we construct in these sections should be viewed intuitively as a ``hard'' input distribution on inputs to the MMPHF problem (in the spirit of Yao's minimax principle).

\subsection{A High-Level Overview of the Proof}\label{sec:overview}

The proof of~\Cref{lem:main} is quite dense and requires both a highly delicate probability distribution and several intricate technical arguments. Thus, before getting into the details of this proof, we 
provide a (very) high-level overview of the logic behind it. In order to convey the intuition, we omit many details from this subsection, instead limiting ourselves to an informal discussion. 

The distribution in~\Cref{lem:main} is roughly as follows: we start with a ``window'' $\Win_1$  which is the interval $[1:M]$, and then sample $X_1$ uniformly at random from $\Win_1$. We then pick window $\Win_2$ to be $[X_1+1:X_1+w_2]$
for an integer $w_2 > 1$ chosen randomly from a carefully designed distribution. Similarly to before, $X_2$ will be chosen uniformly from $\Win_2$. We continue like this by picking a new window $\Win_i = [X_{i-1}+1:X_{i-1}+w_{i}]$ for each $i \in [m]$ by sampling each $w_i$ from a distribution that is constructed based on $(w_1,\ldots,w_{i-1})$, and then sampling $X_i$ from $\Win_i$. Note that, by design, we will satisfy $X_1 < X_2 < \ldots < X_m$.

The key property that this distribution achieves 
can be explained informally as follows. For any index $i \in [m]$, there is a recursive partitioning of the window $\Win_i$ into ``dense'' and ``sparse'' intervals, where an interval $I \subseteq \Win_i$ is dense (with respect to the function $f$ and the index $i$) if at least an $\Omega(1/m)$ fraction of entries $j \in I$ satisfy $f(j) = i$, and otherwise $I$ is sparse. The central property that our distribution ensures is that, if the random choice of $X_i$ places it in a dense interval, then (with very high probability) the \emph{final window} $\Win_m$ will itself end up being dense (i.e., for at least a $2/m$ fraction of $j \in \Win_m$, $f(j) = i$). 

% will contain at least a $2/m$ fractionis the random choice of $X_i$ from $\Win_i$ and majority the points in the last window $\Win_m$ will be chosen from either the same dense interval or the same sparse one. See~\Cref{fig:intervals} below. 

% \bigskip

% \begin{figure}[h!]
% \centering
% \input{intervals}
% \caption{An illustration of the main property established by the distribution of~\Cref{lem:main}. Here, a partitioning of window $\Win_i$ into dense intervals (dashed rectangles) and sparse intervals (empty rectangles) is drawn. A dense interval here is a one such that $f(j) = i$ for $\geq 2/m$ fraction of entries $j$ in the interval. The dot denotes a choice of $X_i$ and the small (filled) rectangle denotes the window $\Win_m$ which most likely intersect primarily with the same dense/sparse interval as $X_i$.}\label{fig:intervals} 
% \end{figure}

% \smallskip

Establishing this property is quite challenging and involves defining the distribution of $w_i$'s in a highly non-uniform manner (in terms of their values); this is also the source of the doubly exponential dependence of  range $M$ on the number of indices $m$. We postpone the details on how this property can be achieved to the actual proof and focus on why it is a useful property for us. 

The analysis of the distribution now uses the property in a potential-function style argument. For each $X_i$, it is either sampled from a sparse interval or a dense one. If $X_i$ is sampled from a sparse interval $I$, then no matter the past iterations, the probability that $f(X_i) = i$ is at most $(2/m)$, since at most $(2/m)$ fraction of $I$ can have value $f(j) = i$ by the definition of it being sparse.  On the other hand, if $X_i$ is chosen from a dense interval, then at least a $(2/m)$ fraction of entries of $\Win_m$ should be mapped to $i$ by $f$ as well (by our property). Seeing $\Win_m$ as a potential function now, we have that this latter step can only happen for $(m/2)$ iterations $i \in [m]$---indeed, each time that this happens for some $i$, we commit some $(2/m)$ fraction of indices $j \in \Win_m$ to having $f(j) = i$, and these sets indices must be disjoint. As a result, we have that only at least $(m/2)$ iterations $i \in [m]$ sample $X_i$ from a sparse interval. Thus, 
\begin{align*}
    \Pr(f(X_1)=1,\ldots,f(X_m)=m) &\leq \hspace{-20pt}\prod_{\substack{\text{$i$: $X_i$ chosen from} \\ \text{a sparse interval}}} \hspace{-20pt}\Pr\paren{f(X_i) = i \mid f(X_1)=1,\ldots,f(X_{i-1})=(i-1)} \\
    & \leq O\left(\frac{1}{m}\right)^{m/2} = m^{-\Omega(m)},
\end{align*}
as desired for the proof of~\Cref{lem:main}. 

The main challenge in formalizing the above argument is the design and analysis of the distribution so that the property discussed above holds. Note also that the property cannot hold deterministically---another challenge is to show that it holds with such high probability that the risk of the property ever failing  (across the entire construction) can be ignored.

% We conclude our overview by mentioning that the above discussion oversimplifies important many details.  For instance, it is not the case that $\Win_m$ is also chosen \emph{uniformly} from the corresponding interval of $X_i$ and we need to show that its distribution will at least be close to uniform. More importantly, this property is not deterministically true of the distribution but we show that it at least holds with a (very) high probability at each index, and argue this is enough for our purpose. Finally, for some indices $i \in [m]$, this property may not even necessarily hold but in those cases, we always have that $X_i$ belongs to a sparse interval which does not involve the potential function argument and thus does not need any guarantee for $\Win_m$ anyway. 

\subsection{The Hard Input Distribution in~\Cref{lem:main}} 

The distribution is defined as follows. 

\begin{tbox}
	\begin{center}
		\underline{The distribution in~\Cref{lem:main}:}
	\end{center}
	\begin{enumerate}[label=$(\roman*)$]
		\item Let $k = m^{m}$, $S_0 = k^{m+1}$, and $X_0 = 0$. 
		\item For $i=1$ to $m$: 
		\begin{enumerate}
			\item Sample two random numbers $Y_i$ from $[2^{S_{i-1}}]$ and $Z_i$ from $[k-1]$ uniformly at random. 
			\item Define the random variables of iteration $i$ as:  
			\[
				X_i = X_{i-1} + Y_i \qquad \text{and} \qquad S_i = S_{i-1} - k^{m-i+1} \cdot Z_i. 
			\]
		\end{enumerate}
		\item Return $(X_1,\ldots,X_m)$ as the resulting random variables. 
	\end{enumerate}
	
\end{tbox}

 To avoid ambiguity, we use lower case letters $(s_i,x_i,y_i,z_i)$ to denote realizations of random variables $(S_i,X_i,Y_i,Z_i)$ for $i \in [m]$. 

We have the following basic observation on the range of numbers created in this distribution. 
\begin{observation}\label{obs:dist}
	Every choice of $(X_1,\ldots,X_m)$ and  $(S_1,\ldots,S_m)$  satisfy the following properties: 
	\begin{enumerate}[label=$(\roman*)$]
		\item \textnormal{Monotonicity:} for all $i \in [m]$, 
		$X_i > X_{i-1}$ and $S_i \leq S_{i-1} - m^{m}$ (and $S_i,X_i$ are integers). 
		\item \textnormal{Boundedness:} for every $i \in [m]$, $X_m \leq X_i + (m-i) \cdot 2^{S_i}$ and $S_{m} \geq S_i - k^{m-i+1} \geq 0$. 
	\end{enumerate}
\end{observation}
\begin{proof}
	Monotonicity of $X_i$'s holds as $Y_i$'s are positive. Monotonicity for $S_i$'s  holds because $Z_i$'s are positive and $k^{m-i+1} \geq k^{m-m+1} \geq k = m^{m}$, meaning that we always have $S_{i} \leq S_{i-1}-m^{m}$.  
	
	\noindent
	For part $(ii)$, we have,
	\[
		X_m = X_i + \sum_{j=i+1}^{m} Y_j \leq X_i + \sum_{j=i+1}^{m} 2^{S_{j-1}} \leq X_i + (m-i) \cdot 2^{S_i}, 
	\]
	which proves the boundedness of $X_i$'s. For $S_i$'s, 
	\begin{align*}
		S_{m} ={S_{i} - \sum_{j=i+1}^{m} {k^{m-j+1}} \cdot Z_j} \geq S_{i} - k^{m} \cdot (k-1) \cdot \sum_{j=i}^{m-1} k^{-j}  \geq S_i - k^{m-i+1}. \tag{as $\sum_{j=i}^{m-1} k^{-j} \leq \sum_{j=i}^{\infty} k^{-j} = k^{-i+1} \cdot (k-1)^{-1}$}
	\end{align*}
	Finally, by this bound, we have $S_m \geq S_0 - k^{m+1} \geq 0$ as $S_0 = k^{m+1}$. 
\end{proof}
\noindent

When discussing $(X_1, \ldots, X_m)$, we will also need some further definitions: 

\begin{itemize}[leftmargin=15pt]
    \item 
For any realization $(s_{<i},x_{<i})$, we define the \textbf{window} of iteration $i \in [m]$, $\Win_i := \Win_i(s_{<i},x_{<i})$, as the support of the random variable $X_i$ conditioned on $(s_{<i},x_{<i})$, i.e., 
\[
	\Win_i := \Win_i(s_{<i},x_{<i}) = [x_{i-1}+1: x_{i-1} + 2^{s_{i-1}}]. 
\]
Notice that $\card{\Win_i(s_{<i},x_{<i})} = 2^{s_{i-1}}$ and $\Win_i$ is determined by $(s_{<i},x_{<i})$.

\item Similarly, for any fixed choice of $(s_{<i},x_{<i})$, consider the following numbers: 
\begin{align}
	{\pw_{i,j} := 2^{s_{i-1} - j \cdot k^{(m-i+1)}} \quad \text{for all $j \in \set{0,\ldots,k}$}}. \label{eq:pw}
%	
%	\pw_{i,0} = 2^{s_{i-1}}, \qquad \pw_{i,1} = 2^{s_{i-1} - \frac{2^{m^{10}}}{k^{i}}}, \qquad \pw_{i,2} = 2^{s_{i-1} - \frac{2^{m^{10}}}{k^i} \cdot 2}, \qquad \cdots \qquad \pw_{i,k} = 2^{s_{i-1} - \frac{2^{m^{10}}}{k^i} \cdot k}.
\end{align}
This way, $\card{\Win_{i+1}(s_{<i},x_{<i})}$ is chosen uniformly at random from $\set{\pw_{i,1},\ldots,\pw_{i,k-1}}$ (depending solely on the choice of $Z_i \in [k-1]$ which also determines $S_i$).  
Moreover, the ratio of $\pw_{i,j}$ and $\pw_{i,j+1}$ is fixed for any $j \in \set{0,\ldots,k-1}$ and we define this quantity as
\begin{align}
	\rw_i := 2^{k^{m-i+1}} = \frac{\pw_{i,j}}{\pw_{i,j+1}} \quad \text{for any $j \in \set{0,\ldots,k-1}$}. \label{eq:rw}
\end{align}
%As before, we use capital letters $\PW_{i,j}$ for $j \in \set{0,\ldots,k}$ and $\pw_i$ to denote the random variables corresponding to $\pw_{i,j}$ and $\pw_i$, respectively.

\end{itemize}

\begin{observation}\label{obs:pw}
	For any fixed $(s_{<i},x_{<i})$, the random variables $\card{\Win_{i+1}},\ldots,\card{\Win_m}$ will be supported on the interval $[2^{m^{m}} \cdot \pw_{i,Z_{i}+1}, \pw_{i,Z_{i}}]$. 
\end{observation}
\begin{proof}
	 By definition, 
	\[
	\card{\Win_{i+1}} = 2^{S_i} = 2^{S_{i-1} - k^{m-i+1} \cdot Z_i} = \pw_{i,Z_i}.
	\] 
	Moreover, by~\Cref{obs:dist}, for any $j \in \set{i+1,\ldots,m}$, we have $\card{\Win_j} \leq \card{\Win_{i+1}}$. Thus each of these windows can have 
	length at most $\pw_{i,Z_i}$, proving the upper bound side.
	
	For the lower bound, for any $j \in \set{i+1,\ldots,m}$, we have, 
	\begin{align*}
	\card{\Win_j} &\geq \card{\Win_{m}} =2^{S_{m-1}} \geq 2^{S_i-k^{m-i+1} + m^{m}} \tag{by part $(ii)$ of~\Cref{obs:dist}} \\
	&= 2^{m^{m}} \cdot 2^{S_i} \cdot 2^{-k^{m-i+1}} = 2^{m^{m}} \cdot \pw_{i,Z_i} \cdot \rw_i^{-1} = 2^{m^{m}} \cdot \pw_{i,Z_i+1}.
	\end{align*}
	This concludes the proof. 
\end{proof}
\noindent 
We need one final definition for now: 

\begin{itemize}[leftmargin=15pt]
    \item 
For the function $f: [M] \rightarrow [m]$, we define the \textbf{density} of index $i \in [m]$ in $f$ over a window $\Win$, denoted by $\dense_f(\Win,i)$, as 
\[
	\dense_f(\Win,i) := \frac{\card{\set{j \in \Win: f(j) = i}}}{\card{\Win}},
\]
namely, the fraction of entries of the window that are equal to $i$. 

\end{itemize}
\begin{observation}\label{obs:density-sample}
	For any choice of $(s_{<i},x_{<i})$, we have, 
	\[
		\Pr\paren{f(X_i) = i \mid s_{<i},x_{<i}} = \dense_f(\Win_i(s_{<i},x_{<i}),i). 
	\]
\end{observation}
\begin{proof}
	Conditioned on $(s_{<i},x_{<i})$, $X_i$ is chosen uniformly at random from $\Win_i(s_{<i},x_{<i})$. The observation therefore follows from the definition of $\dense_f(\Win_i(s_{<i},x_{<i}),i).$
\end{proof}

\subsection{Analysis of the Hard Distribution -- Proof of~\Cref{lem:main}} \label{sec:analysis}

We prove~\Cref{lem:main} by individually considering each iteration in the distribution. 

\todo{sepehr:check from here till the end nothing breaks after the change} 
\begin{lemma}\label{lem:iteration}
	For any iteration $i \in [m]$ and conditioned on any choice of $(s_{<i},x_{<i})$, \underline{at least one} of the following two conditions is true: 
	\[
		(i)~\Pr\left(f(X_i) = i  \mid  s_{<i},x_{<i} \right) \leq \frac{101}{m} \quad \textnormal{or} \quad (ii)~\Pr\left(\dense_f(\Win_m,i) < \frac{2}{m}  \mid s_{<i},x_{<i} \right) < \frac{1}{k^{1/3}}.
	\]
%Moreover, this holds even if we condition on any fixed outcome for the random bits that determine $Y_{i + 1}, Z_{i + 1}, Y_{i + 2}, Z_{i + 2}, \ldots, Y_m, Z_m$. 
\end{lemma}

The main bulk of this section is to prove~\Cref{lem:iteration}. We then show at the end of the section that this lemma easily implies~\Cref{lem:main}. To continue, we need some definitions.

\begin{definition}\label{def:window-tree}
	The {\textnormal{\textbf{window-tree}}} of iteration $i \in [m]$ for $(s_{<i},x_{<i})$, denoted by $\TT_i := \TT(s_{<i},x_{<i})$, is the following rooted tree with $k+1$ levels (the root is at level $0$):
	\begin{enumerate}[label=$(\roman*)$]
		\item Every non-leaf node $\alpha$ of the tree has $\rw_i$ many child-nodes. 
		\item Every node $\alpha$ at a level $\ell \in \set{0,\ldots,k}$ is associated with a window $\Win(\alpha)$ of length $\pw_{i,\ell}$. 
		\item The root $\alpha_r$ is associated with the window $\Win(\alpha_r) := \Win_i(s_{<i},x_{<i})$. The windows associated with child-nodes of a node $\alpha$ at level $\ell$ partition $\Win(\alpha)$ of length $\pw_{i,\ell}$ into equal-size windows of length $\pw_{i,\ell+1}$ (recall that $\alpha$ has $\rw_i = \pw_{i,\ell}/\pw_{i,\ell+1}$ child-nodes). Moreover, the left most child-node receives the window in the partition with the smallest starting point, the next child-node on the right receives the next window
		with smallest part, and so on. 
		\item The \textnormal{\textbf{density}} of a node $\alpha$ with respect to any function $f: [M] \rightarrow [m]$ is defined as 
		\[
		\dense_f(\alpha) := \dense_f(\Win(\alpha),i).
		\] 
	\end{enumerate}
\end{definition}

One way we use the window-tree in our analysis is to consider the process of sampling $X_i$ (which is uniform over $\Win_i(s_{<i},x_{<i})$ at this stage) as {traversing} the window-tree via a root-to-leaf path. 
This is formalized in the following observation. 

\begin{observation}\label{obs:sampling-tree}
	The distribution of $X_i$ conditioned on $(s_{<i},x_{<i})$ can be alternatively seen as: $(i)$ Sample a root-to-leaf path $\alpha_0,\alpha_1,\ldots,\alpha_k$ where $\alpha_0$ is the root of $\TT_i$ and where each $\alpha_{\ell+1}$ 
	is a child-node of $\alpha_\ell$ chosen uniformly at random; then, $(ii)$ sample $X_i$ uniformly at random from $\Win(\alpha_k)$. We refer to $\alpha_0,\ldots,\alpha_k$ as the \textnormal{\textbf{sampling path}} of $X_i$. 
\end{observation}
\begin{proof}
	$X_i$ is distributed uniformly over $\Win_i$ and leaf-nodes of $\TT_i$ form an equipartition of $\Win_i$. 
\end{proof}

In addition, we define a {pruning} procedure for any window-tree $\TT$ as follows. 

\begin{definition}
	Fix a function $f: [M] \rightarrow [m]$ and a window-tree $\TT_i$ for some $i \in [m]$. We say that a node $\alpha \in \TT_i$ is \textnormal{\textbf{sparse}} iff 
	\[
		\dense_f(\alpha) \leq \frac{100}{m}. 
	\]
	We have the following procedure for {pruning} $\TT_i$: Start from the root down to the leaf-nodes and prune any sparse node of the tree, as well as \underline{all} of that node's sub-tree. 
	We refer to a sparse node that was pruned on its own (i.e., any node that is sparse and has no sparse ancestors) as a \textnormal{\textbf{directly pruned}} node and to other pruned nodes (i.e., nodes with sparse ancestors) as \textnormal{\textbf{indirectly pruned}}.

	Finally, for $\ell \in \set{0,\ldots,k}$, define $p_\ell$ as the 
	fraction of directly pruned nodes at level $\ell$ of the tree over all level-$\ell$ nodes that are \underline{not} indirectly pruned. 
\end{definition}

It is worth noting that pruning is deterministic  conditioned on $(s_{<i},x_{<i})$.

With these definitions, we can now start proving~\Cref{lem:iteration}. This will be done by considering some different cases handled by the following claims. The first (and easiest) case is when most nodes of the window-tree are pruned, in 
which case we achieve property $(i)$ of~\Cref{lem:iteration}.  

\begin{claim}[Case I: ``Many Directly Pruned Nodes'']\label{clm:case1}
	Suppose 
	\[
	\prod_{\ell=0}^{k} (1-p_\ell) \leq \frac{1}{m}.
	\]
	Then, for any choice of $(s_{<i},x_{<i})$, 
	\[
		\Pr_{X_i}(f(X_i) = i  \mid s_{<i},x_{<i}) \leq \frac{101}{m}.
	\]
\end{claim}
\begin{proof}
	Let $W_{\text{rem}}$ denote the subset of $\Win_i$ that remains after removing windows of all pruned leaf-nodes from $\Win_i$. We have that
	\[
		\card{W_{\text{rem}}} = \frac{\text{\# leaf-nodes of $\TT_i$ that are not pruned}}{\text{\# leaf-nodes of $\TT_i$}} \cdot \card{\Win_i} = \prod_{\ell=0}^{k} (1-p_\ell) \cdot \card{\Win_i} \leq \frac{\card{\Win_i}}{m},
	\]
	where the second equality is because at each level $\ell$ of the tree, the number of not pruned nodes drops by a factor of $(1-p_\ell)$ by the definition of $p_\ell$.

	Let $\DP$ denote the set of all nodes in the tree $\TT_i$ that were directly pruned. 
	Note that the windows $\Win(\alpha)$ for $\alpha \in \DP$ partition $\Win_i \setminus W_{\text{rem}}$. This implies that 
	\begin{align*}
		\dense_f(\Win_i,i) &= \frac{1}{\card{\Win_i}} \cdot \paren{\card{W_{\text{rem}}} \cdot \dense_f(W_{\text{rem}},i) + \sum_{\alpha \in \DP} \dense_f(\alpha) \cdot \card{\Win(\alpha)}}  \tag{by the definition of $\dense_f(\cdot)$ function} \\
		&\leq \frac{1}{\card{\Win_i}} \cdot \paren{\card{W_{\text{rem}}} + \sum_{\alpha \in \DP} \frac{100}{m} \cdot \card{\Win(\alpha)}} \tag{as $\dense_f(\alpha) \leq 100/m$ by the definition of sparsity, and $\dense_f(W_{\text{rem}},i) \leq 1$} \\
		&\leq \frac{1}{m} + \frac{100}{m} = \frac{101}{m}. \tag{as $\card{W_{\text{rem}}}/\card{\Win_i} \leq 1/m$ as established above, and $\sum_{\alpha \in \DP} \card{\Win(\alpha)} \leq \card{\Win_i}$}
	\end{align*}
	By~\Cref{obs:density-sample}, we have, 
	\[
		\Pr_{X_i}(f(X_i) = i  \mid s_{<i},x_{<i}) = \dense_f(\Win_i,i) \leq \frac{101}{m}, 	
	\]
	concluding the proof. 
\end{proof}

We now consider the complementary case, while also taking the randomness of $Z_i$ into account. Recall that 
$Z_i$ is uniform over $[k-1]$ and that $\card{\Win_{i+1}} = \pw_{i,Z_i}$. For any fixed realization $z_i$ of $Z_i$,  recall the sampling-path-based process of sampling $X_i$ outlined in~\Cref{obs:sampling-tree}. 
Consider the first $z_i$ vertices in this path, namely, $\alpha_0,\ldots,\alpha_{z_{i}-1}$ that start from the root and end at a level $z_{i}-1$ node of $\TT_i$.

Define Event $\event(s_{<i},x_{<i},z_i,X_i)$ to be the event that none of the nodes in $\alpha_0,\ldots,\alpha_{z_{i}-1}$ are pruned.
Event $\event(s_{<i},x_{<i},z_i,X_i)$ depends only on the choice of $X_i$ (to traverse the root-to-leaf path), and is conditioned on $s_{<i},x_{<i}$ (which determine the window-tree $\TT_i$) and $z_i$ (which determines the level of the tree that we focus on). To avoid clutter, when it is clear from the context, we refer to this event simply by $\event_i$. 

We partition the remaining cases based on whether or not the event $\event_i$ happens. 

\begin{claim}[Case II: ``A Pruned Node on the Sampling Path'']\label{clm:case2}
	Fix any choice of $z_i$ and $(s_{<i},x_{<i})$. In the case that the event $\event_i$ does \underline{not} happen, we have, 
	\[
		\Pr_{X_i}(f(X_i) = i  \mid s_{<i},x_{<i}, z_i, \overline{\event(s_{<i},x_{<i},z_i,X_i)}) \leq \frac{100}{m}.
	\]
%	(Again, note that conditioned on $x_{<i}$ obtained from $y_{<i}$, $Y_i$ uniquely determines $X_i$.)
\end{claim}
\begin{proof}
	After conditioning on $(s_{<i},x_{<i},z_i)$, the event $\event_i$ is only a function of the sampling process of $X_i$  outlined in~\Cref{obs:sampling-tree}. Assuming $\event_i$ does not happen, we know that there exists
	a \emph{unique} node $\alpha_j$ on the path $\alpha_0,\ldots,\alpha_{z_{i}-1}$ such that $\alpha_j$ is sparse and is directly pruned. By additionally conditioning on the subpath $\alpha_0,\ldots,\alpha_j$, 
	we have that $X_i$ is chosen uniformly at random from $\Win(\alpha_j)$ at this point. Thus, 
	\begin{align*}
		&\Pr_{X_i}(f(X_i) = i  \mid s_{<i},x_{<i}, z_i, \overline{\event_i}) \\
		&= \sum_{j=0}^{z_i-1}\hspace{-15pt}\sum_{\substack{(\alpha_1,\ldots,\alpha_j): \\ \text{$\alpha_j$ is directly pruned}}} 
		\hspace{-20pt}\Pr_{X_i}(f(X_i)=i \wedge \text{$(\alpha_1,\ldots,\alpha_j)$ is on the sampling path} \mid s_{<i},x_{<i}, z_i, \overline{\event_i}) \tag{as these subpaths partition all possible choices for $\event_i$ to not happen}\\
		&= \sum_{j=0}^{z_i-1}\hspace{-15pt}\sum_{\substack{(\alpha_1,\ldots,\alpha_j): \\ \text{$\alpha_j$ is directly pruned}}} 
		\hspace{-20pt}\Pr_{X_i}(\text{$(\alpha_1,\ldots,\alpha_j)$ is on the sampling path} \mid s_{<i},x_{<i}, z_i, \overline{\event_i}) \cdot \frac{\card{\set{t \in \Win(\alpha_j): f(t) = i}}}{\card{\Win(\alpha_j)}} \tag{as $X_i$ is chosen uniformly from $\Win(\alpha_j)$ under these conditions} \\
		&= \sum_{j=0}^{z_i-1}\hspace{-15pt}\sum_{\substack{(\alpha_1,\ldots,\alpha_j): \\ \text{$\alpha_j$ is directly pruned}}} 
		\hspace{-20pt}\Pr_{X_i}(\text{$(\alpha_1,\ldots,\alpha_j)$ is on the sampling path} \mid s_{<i},x_{<i}, z_i, \overline{\event_i}) \cdot \dense_f(\Win(\alpha_j),i) \tag{by the definition of $\dense_f$} \\
		&\leq  \sum_{j=0}^{z_i-1}\hspace{-15pt}\sum_{\substack{(\alpha_1,\ldots,\alpha_j): \\ \text{$\alpha_j$ is directly pruned}}} 
		\hspace{-20pt}\Pr_{X_i}(\text{$(\alpha_1,\ldots,\alpha_j)$ is on the sampling path} \mid s_{<i},x_{<i}, z_i, \overline{\event_i}) \cdot \frac{100}{m} \tag{as $\alpha_j$ needs to be sparse to be directly pruned}.
		%\Pr_{X_i}(f(X_i)=i \mid \text{$X_i$ is from a sparse $\alpha_j$}) = \frac{\card{\set{t \in \Win(\alpha_j): f(t) = i}}}{\card{\Win(\alpha_j)}} = \dense_f(\Win(\alpha_j),i) \leq \frac{100}{m}. 
	\end{align*}
	This can now be further upper bounded by $100/m$ as the probability terms are summing over all disjoint events that can lead to $\overline{\event_i}$ (conditioned on this event) and thus add up to one. 
\end{proof}

Finally, we have the following case which handles the situation when $\event_i$ happens. The following claim is the heart of the proof.

\begin{claim}[Case III: ``No Pruned Nodes on the Sampling Path'']\label{clm:case3}
	Fix any choice of $z_i$ and $(s_{<i},x_{<i})$. In the case that the event $\event_i$ happens, we have, 
	\[
		\Pr_{X_i}\left(\dense_f(\Win_m,i) < \frac{2}{m}  \mid s_{<i},x_{<i},z_i, \event(z_i,X_{i})\right)  < 4 \cdot \left(p_{z_i} + p_{z_i+1} + \frac{m}{r_i}\right). 
	\]
\end{claim}
\begin{proof}
	Throughout this proof, we always condition on $s_{<i},x_{<i},z_i, \event(z_i,X_{i})$ and thus may not mention this explicitly in the probability terms. 
	Let us first list the information we have so far: 
	\begin{itemize}
		\item The node $\alpha_{z_{i}-1}$ on the sampling path is not pruned as we conditioned on the event $\event(z_i,X_i)$ (we emphasize that $\alpha_{z_i-1}$ is a random variable and is not fixed yet just by these conditions). 
		\item Window $\Win_m$ is going to have size at least $2^{m^{m}} \cdot \pw_{i,z_{i}+1}$ and at most $\pw_{i,z_{i}}$ by~\Cref{obs:pw}. 
		\item By~\Cref{obs:dist},
		\[
			X_m \leq X_i + (m-i) \cdot 2^{S_i} = X_i + (m-i) \cdot \pw_{i,z_i}. \tag{by the definition of $\pw_{i,z_i} = 2^{S_i}$}
		\]
		\item $\Win_m$ starts at $X_m$ and ends at $X_m + \card{\Win_m}$. We can think of the process of sampling $\Win_m$  
		as first sampling its length $\card{\Win_m}$, then sampling the \textbf{offset} $O_{i,m} := X_m - X_i = \sum_{j=i+1}^{m} Y_j$ conditioned on $\card{\Win_m}$, and then sampling $X_i$ conditioned on $O_{i,m}$, and $\card{\Win_m}$.  
		
		\item We further have that  $X_i$ conditioned on $O_{i,m}$ and $\card{\Win_m}$ is still uniform over $\Win(\alpha_{z_i-1})$. This is because $\card{\Win_m}$ is only a function of $Z_{i+1},\ldots,Z_m$, 
		and $X_m - X_i$ is only a function of $Y_{i+1},\ldots,Y_{m}$, while $X_i$ is only a function of $Y_i$; finally, $Y_i$ is independent of $Y_{i+1},\ldots,Y_m$ and $Z_{i+1},\ldots,Z_m$ and 
		is chosen uniformly from $[2^{s_{i-1}}]$. 
	\end{itemize}
	\noindent
	In the following, we condition on any fixed choice of offset $o_{i,m}$ for $O_{i,m}$ and on $\card{\Win_m}$. We have already established that
	\begin{align}
		2^{m^{m}} \cdot\pw_{i,z_i+1} \leq \card{\Win_m} \leq \pw_{i,z_i} \qquad \text{and} \qquad o_{i,m} \leq (m-i) \cdot \pw_{i,z_i}. \label{eq:final-condition}
	\end{align}
	
	Moreover, the distribution of $\Win_m$ conditioned on $o_{i,m},\card{\Win_m}$ (and $s_{<i},x_{<i},z_i,\event_i$ that we always condition on in this proof), 
	is $X_i + o_{i,m}$ for $X_i$ chosen randomly from $\Win(\alpha_{z_i-1})$. Moreover, 
	 given that $o_{i,m} \leq (m-i) \cdot \pw_{i,z_i}$ while $\card{\Win(\alpha_{z_i-1})} = \pw_{i,z_i-1} = r_i \cdot \pw_{i,z_i}$ and $r_i = 2^{k^{m-i+1}} \geq 2^{k}$ as $i \leq m$, the distribution of 
	 $X_i$ and $X_i + o_{i,m}$ are quite close to each other modulo a negligible factor. Thus, 
	 for intuition,  we can think of $X_i$ itself as the distribution of starting point for $\Win_m$ in this context (although we will of course take this difference into account explicitly in the proof).  
	 We now use this information to prove the claim. To simplify the exposition, we are going to separate the analysis based on level $z_i$ and level $z_{i+1}$ of the window-tree. 
	 
	 \paragraph{Analysis on level $z_i$ of the window-tree.} Firstly, since $\card{\Win_m} \leq \pw_{i,z_i}$, and each node at level $z_i$ of the window-tree $\TT_i$ has a window of length $\pw_{i,z_i}$, we get that $\Win_m$ intersects with windows of 
	at most two \emph{consecutive} nodes at level $z_{i}$ of $\TT_i$, which are solely determined by the choice of $X_i$. We use $\beta_1(X_i)$ and $\beta_2(X_i)$ to denote these two nodes with $\beta_1$ 
	being the one where the starting point of $\Win_m$, namely, $X_i + o_{i,m}$, lies in, and $\beta_2(X_i)$ being the one containing the endpoint $X_i + o_{i,m} + \card{\Win_m}$ (note that it is possible that $\beta_2=\beta_1$).

	We prove that with high probability, neither of these nodes are pruned. Let us focus on $\beta_1(X_i)$ first (the analysis is almost identical for $\beta_2(X_i)$ and we can then apply the union bound).  
	For any $\ell \in \set{0,\ldots,k-1}$, let $\PP(\ell)$ (resp. $\DP(\ell)$) denote the set of pruned (resp. directly pruned) nodes at level $\ell$ of $\TT_i$; similarly, for a node $\alpha \in \TT_i$, let $\PP(\alpha)$ (resp. $\DP(\alpha)$) denote the set of child-nodes of $\alpha$ that are pruned (resp. directly pruned).  
	For any fixed choice of $\alpha_{z_i-1}$ on the sampling path of $X_i$,
	%For any fixed choice of $\alpha_{z_i-1}$ on the sampling path of $X_i$, we have,  
	\begin{align}
%		\Pr_{X_i}\paren{\text{$\beta_1(X_i)$ is pruned}} &= \Exp_{\alpha_{z_i-1}} {\Pr_{X_i}\paren{\text{$\beta_1(X_i)$ is pruned} \mid \alpha_{z_i-1}}} \tag{by law of total probability, over the choice of $\alpha_{z_i-1}$ in the sampling path} \\
		\Pr_{X_i}\paren{\text{$\beta_1(X_i)$ is pruned} \mid \alpha_{z_i-1}} 	&= \sum_{\beta \in \PP(z_{i})}\Pr_{X_i}\paren{\beta_1(X_i) = \beta \mid \alpha_{z_{i-1}}} \tag{as $\beta_1$ is in level $z_i$ and $\PP(z_i)$ is the set of all pruned nodes of this level} \\
		&= \sum_{\beta \in \PP(\alpha_{z_i-1})}\Pr_{X_i}\paren{\beta_1(X_i) = \beta \mid \alpha_{z_{i-1}}} +  
		\hspace{-10pt}\sum_{\substack{\beta \in \\ \PP(z_i) \setminus \PP(\alpha_{z_i-1})}}\hspace{-10pt}\Pr_{X_i}\paren{\beta_1(X_i) = \beta \mid \alpha_{z_{i-1}}} 
		\tag{by partitioning the nodes in level $z_i$ between child-nodes of $\alpha_{z_i-1}$ and remaining ones} \\
		&\leq \card{\PP(\alpha_{z_i-1})} \cdot \frac{1}{{r_i}} + (m-i) \cdot \frac{1}{{r_i}}, \label{eq:has-2-parts}
	\end{align}
	where the last inequality holds because of the following reasoning. Firstly, the probability that $\beta_1(X_i)$ is equal to any fixed node $\beta$ at level $z_{i}$ is at most $1/r_i$. This is because 
	\[
		\Pr\paren{\beta_1(X_i) = \beta \mid \alpha_{z_{i-1}}} = \Pr\paren{X_i+o_{i,m} \in \Win(\beta) \mid \alpha_{z_i-1}} \leq \frac{\card{\Win(\beta)}}{\card{\Win(\alpha_{z_{i}-1})}} = \frac{1}{r_i},
	\]
	because $X_i$ is chosen uniformly from $\Win(\alpha_{z_i-1})$, and $\card{\Win(\beta)} = \card{\Win(\alpha_{z_i-1})}/r_i$ as $\beta$ is at level $z_i$. This immediately implies the first term in the RHS 
	of~\Cref{eq:has-2-parts}. For the second term, for $\beta_1(X)$ to intersect with a node $\beta$ not in the subtree of $\alpha_{z_i-1}$, we need to have $X_i + o_{i,m} \notin \Win(\alpha_{z_i-1})$, 
	while we know $X_i \in \Win(\alpha_{z_i-1})$. As $o_{i,m} \leq (m-i) \cdot \pw_{i,z_i}$ by~\Cref{eq:final-condition}, and any node at level $z_i$ has a window of length $\pw_{i,z_i}$, 
	we get that there are most $(m-i)$ choices of $\beta$ outside child-nodes of $\alpha_{z_i-1}$ that can also become $\beta_1(X_i)$. The second part of RHS in~\Cref{eq:has-2-parts} 
	now follows from this and the upper bound of $1/r_i$ on the probability of each node. 
	
	Finally, by taking the expectation over the choice of $\alpha_{z_i-1}$, 
	\begin{align*}
		\Pr_{X_i}\paren{\text{$\beta_1(X_i)$ is pruned}} &= \Exp_{\alpha_{z_i-1}} \bracket{\Pr_{X_i}\paren{\text{$\beta_1(X_i)$ is pruned} \mid \alpha_{z_i-1}}} \tag{by the law of total probability, over the choice of $\alpha_{z_i-1}$ in the sampling path} \\
		&\leq \Exp_{\alpha_{z_i-1}} \bracket{\frac{\card{\PP(\alpha_{z_i-1})}}{r_i}} + \frac{(m-i)}{r_i} \tag{by~\Cref{eq:has-2-parts}} \\
		&= p_{z_i} + \frac{(m-i)}{r_i}, 
	\end{align*}
	where in the final equality, we used the fact that $\alpha_{z_i-1}$ is chosen from non-pruned nodes (by conditioning on $\event_i$), and thus $\card{\PP(\alpha_{z_i-1})}/r_i$ is the fraction of pruned nodes over all \emph{not} indirectly pruned at level $z_i$, 
	which by definition is $p_{z_i}$. 
	
	Doing the same exact analysis, we can bound the probability that $\beta_2(X_i)$ is pruned also as 
	\[
		\Pr_{X_i}\paren{\text{$\beta_2(X_i)$ is pruned}} \leq p_{z_i} + \frac{(m-i)+1}{r_i},
	\]
	where the $+1$ term in the RHS compared to the one for $\beta_1$ comes from the fact that $\beta_2(X_i)$ can have $(m-i+1)$ choices outside subtree of $\alpha_{z_i-1}$ (because we are now considering 
	$X_i + o_{i,m} + \card{\Win_m} \leq X_i + (m-i+1) \cdot \pw_{i,z_i}$ instead). By the union bound on the probabilities for $\beta_1(X_i)$ and $\beta_2(X_i)$, 
	\begin{align}
		\Pr_{X_i}\paren{\text{either of $\beta_1(X_i)$ or $\beta_2(X_i)$ is pruned}} \leq 2 \cdot p_{z_i} + 2 \cdot \frac{m}{r_i}. \label{eq:level-zi}
	\end{align}
	
	\paragraph{Analysis on level $z_i+1$ of the window-tree.} For the next step, let $\gamma_1(X_i),\ldots,\gamma_t(X_i)$ denote the child-nodes of $\beta_1(X_i)$ and $\beta_2(X_i)$ 
	such that $\Win(\gamma_j(X_i))$ is \emph{entirely} contained in $\Win_m$. Again, the choice of $\gamma_1,\ldots,\gamma_t$ is only a function of $X_i$. Moreover, 
	since $\card{\Win_m} \geq 2^{m^{m}} \cdot \pw_{i,z_i+1}$ by~\Cref{eq:final-condition}, while the window of each node at level $z_i+1$ is of size $\pw_{i,z_i+1}$, we have that $t \geq 2^{m^{m}}-2$ always. 
	We now bound the probability that each $\gamma_j$ is (directly) pruned, for $j \in [t]$. This part of the analysis is quite similar to that of level $z_i$ with only minor changes. 
	
	For any choice of $\beta_1(X_i)$ and $\beta_2(X_i)$, 
	\begin{align}
		\Pr_{X_i}\paren{\text{$\gamma_j(X_i)$ is directly pruned} \mid \beta_1,\beta_2} &= \hspace{-15pt}\sum_{\substack{\gamma \in \\\DP(\beta_1) \cup \DP(\beta_2)}}\hspace{-15pt}\Pr_{X_i}\paren{\gamma_j(X_i) = \gamma \mid \beta_1,\beta_2} 
		\tag{because $\Win_m \subseteq \Win(\beta_1) \cup \Win(\beta_2)$ and thus $\gamma_j$ has no choice outside child-nodes of $\beta_1$ or $\beta_2$}\\ 
		&\leq \Paren{\card{\DP(\beta_1)} + \card{\DP(\beta_2)}}\cdot \frac{1}{r_i},  \label{eq:has-one-part}
	\end{align}
	where we are again going to argue that the probability that $\gamma_j(X_i)$ is equal to any fixed node $\gamma$ is at most $1/r_i$ conditioned on the choice of $\beta_1$ and $\beta_2$. 
	For $\gamma_j(X_i)$ to be equal to a node $\gamma$ we need to have that $X_i + o_{i,m} + (j-1) \cdot \pw_{i,z_i+1} \in \Win(\gamma)$; this is because $\gamma_j(X_i)$ appears after $(j-1)$ nodes of level $z_{i}+1$ that are fully inside 
	$\Win_m$ and each such window has length $\pw_{i,z_i+1}$ (note that this is a necessary but not a sufficient condition). Thus, 
	\[
		\Pr_{X_i}\paren{\gamma_j(X_i) = \gamma \mid \beta_1,\beta_2} \leq 	\Pr_{X_i}\paren{X_i + o_{i,m} + (j-1) \cdot \pw_{i,z_i+1} \in \Win(\gamma) \mid \beta_1,\beta_2} \leq \frac{\card{\Win(\gamma)}}{\pw_{i,z_i}} = \frac{1}{r_i},
	\]
	where the last inequality is because conditioned on $\Win_m$ intersecting with $\beta_1,\beta_2$, $X_i$ is chosen uniformly at random from a window of length $\pw_{i,z_i}$ (equal to length of $\Win(\beta_1)$ and $\Win(\beta_2)$); the final
	equality also uses that $\card{\Win(\gamma)} = \pw_{i,z_i+1} = \pw_{i,z_i}/r_i$. Hence \eqref{eq:has-one-part}.
	
% 	Combining~\Cref{eq:has-one-part} and linearity of expectation, we have that
% 	\[
% 		\Exp_{X_i}\bracket{\text{\# of $\gamma_1(X_i),\ldots,\gamma_t(X_i)$ that are pruned} \mid \beta_1,\beta_2} \leq \Paren{\card{\PP(\beta_1)} + \card{\PP(\beta_2)}} \cdot \frac{t}{r_i}. 
% 	\]
	We can now deduce that
	\begin{align}
		&\Exp_{X_i}\bracket{\text{\# of $\gamma_1(X_i),\ldots,\gamma_t(X_i)$ that are directly pruned}} \notag \\
		 &\hspace{10pt}=\Exp_{\beta_1,\beta_2} \Exp_{X_i}\bracket{\text{\# of $\gamma_1(X_i),\ldots,\gamma_t(X_i)$ that are directly pruned} \mid \beta_1,\beta_2} \tag{by the law of total probability over the choices of $\beta_1,\beta_2$ 
		 } \\
		 &\hspace{10pt} = \Exp_{\beta_1,\beta_2} \bracket{\card{\DP(\beta_1)} + \card{\DP(\beta_2)}
		 \cdot \frac{t}{r_i}}, \label{eq:late-night} 
    \end{align}
    where the last inequality is by \Cref{eq:has-one-part}. 
    
    % Recall that $p_{z_i + 1}$ is the fraction of directly pruned nodes in level $z_i + 1$, out of all nodes in level $z_i + 1$ that are not indirectly pruned. If we 
    
    Let $\overline{\PP}(z_i)$ denote the set of not pruned nodes in level $z_i$ and let $\hat{\PP}(z_i)$ denote the set of nodes in level $z_i$ whose parents are not pruned. Since we are conditioning on $\event_i$, we know that $X_i$ is uniformly random from the interval $\cup_{\beta \in \hat{\PP}(z_i)}\Win(\beta)$. It follows that $X_m = X_i + o_{i,m}$ is uniformly random in a range whose size is also $\ell = \sum_{\beta \in \hat{\PP}(z_i)}|\Win(\beta)|$. Thus, for any level-$z_i$ node $\beta$, we have that 
    \[
    \Pr[\beta_1 = \beta] = \Pr[X_m \in \Win(\beta)] \leq \frac{\card{\Win(\beta)}}{\ell} = \frac{\pw_{i,z_i}}{\ell} = \frac{1}{|\hat{\PP}(z_i)|}.
    \]
    Summing over the level-$(z_i + 1)$ nodes that are directly pruned, we have that
    \[
    \E|\DP(\beta_1)| = \sum_{\gamma \in \DP(z_i + 1)} \Pr[\beta_1 \text{ is the parent of } \gamma] \le \frac{|\DP(z_i + 1)|}{|\hat{\PP}(z_i)|} \le \frac{|\DP(z_i + 1)|}{|\overline{\PP}(z_i)|} ,
    \]
    using the upper bound established above on the probability that $\beta_1$ is any fixed node. 
    Note that 
    \[
    p_{z_i + 1} = \frac{|\DP(z_i + 1)|}{r_i \cdot  |\overline{\PP}(z_i)|},
    \]
    i.e.,   the number of directly pruned nodes in level $z_i + 1$ divided by the number of nodes with not pruned parents. Therefore,  $\Exp{\card{\DP(\beta_1)}} \leq r_i \cdot p_{z_i + 1}$. By the same reasoning (but applied to $\beta_2$, which contains the endpoint of $X_m$), we have that $\Exp{\card{\DP(\beta_2)}} \leq r_i \cdot p_{z_i + 1}$. 
    
    Thus, we can use \Cref{eq:late-night} to conclude that
    \[\Exp_{X_i}\bracket{\text{\# of $\gamma_1(X_i),\ldots,\gamma_t(X_i)$ that are directly pruned}} \leq 2 p_{z_i + 1} \cdot t. \]

	By Markov's inequality, 
	\begin{align}
		\Pr_{X_i}\paren{\text{more than $t/2$ of $\gamma_1(X_i),\ldots,\gamma_t(X_i)$ are directly pruned}} \leq 4 \cdot p_{z_i+1}. \label{eq:has-one-part-n}
	\end{align}
	Finally, by considering the possibility that at least one of $\beta_1$ or $\beta_2$ could be pruned also we have, 
	\begin{align}
		&\Pr_{X_i}\paren{\text{more than $t/2$ of $\gamma_1,\ldots,\gamma_t$ are pruned}} \notag \\
		&\hspace{30pt} \leq \Pr_{X_i}\paren{\text{more than $t/2$ of $\gamma_1(X_i),\ldots,\gamma_t(X_i)$ are directly pruned}} \notag \\
		&\hspace{60pt} + \Pr(\text{either of }\beta_1 \text{ or }\beta_2\text{ are pruned})\notag \\
		&\hspace{30pt} \leq 4  p_{z_i+1} + 2p_{z_i} + \frac{2m}{r_i},  \label{eq:level-zi+1}
	\end{align}
	by~\Cref{eq:level-zi} and~\Cref{eq:has-one-part-n}. 
	
\paragraph{Concluding the proof.} Let us now condition on the event that at least $t/2$ of nodes $\gamma_1,\ldots,\gamma_t$ are \emph{not} pruned, namely, the complement of the event in~\Cref{eq:level-zi+1}. 
Given that $\Win_m$ can have intersection with at most two other level-$(z_{i}+1)$ nodes beside $\gamma_1,\ldots,\gamma_t$, conditioned on the above event, we have, 
	\[
		\dense_f(\Win_m,i) \geq \frac{(t/2) \cdot 100/m}{t+2}  \geq \frac{100}{3m} >  \frac{2}{m}, 
	\]	
	as $t \geq 2^{m^m}-2 \gg 1$. Thus, by~\Cref{eq:level-zi+1}, we have, 
	\[
		\Pr_{X_i}\paren{\dense_f(\Win_m,i) \leq \frac{2}{m}} \leq 2p_{z_i} + 4 \cdot p_{z_i+1} + \frac{2m}{r_i} < 4 \left(p_{z_i} + p_{z_i+1} + \frac{m}{r_i}\right), 
	\]
	concluding the proof. 
\end{proof}

\Cref{clm:case1,clm:case2,clm:case3} now cover all possible cases and allow us to prove~\Cref{lem:iteration}.

\begin{proof}[Proof of~\Cref{lem:iteration}]
		Fix the tree $\TT_i$ and consider its pruning process. If $\prod_{\ell=0}^{k} (1-p_\ell) \geq 1/m$, we achieve the first condition of the lemma by~\Cref{clm:case1} and are thus done. Now consider
		the complement case. In this case, we have, 
		\[
			\frac{1}{m} < \prod_{\ell=0}^{k} (1-p_\ell) \leq \exp\paren{-\sum_{\ell=0}^{k} p_\ell},
		\]
		which implies that $	\sum_{\ell=0}^{k} p_\ell \leq \ln{m}$. Recall that the choice of $Z_i$ in the distribution is uniform over $[k-1]$ regardless of conditioning on $(s_{<i},x_{<i})$. 
		Thus, we have, 
		\[
			\Exp_{Z_i}\bracket{p_{Z_i}+p_{Z_{i}+1}} \leq \frac{1}{k-1} \cdot \sum_{\ell=1}^{k-1} p_{\ell} + \frac{1}{k-1} \cdot \sum_{\ell=2}^{k} p_\ell \leq \frac{2}{k-1} \sum_{\ell=0}^{k} p_\ell \leq \frac{2\ln{m}}{(k-1)}. 
		\]
		By Markov's inequality, we have, 
		\[
			\Pr_{Z_i}\paren{p_{Z_i}+p_{Z_i+1} \geq \frac{4\cdot\ln{m}}{k^{1/2}}} \ll \frac{1}{k^{1/3}}. 
		\]
		We can now condition on any choice $z_i$ of $Z_i$ such that $p_{z_i} + p_{z_i+1} \leq (4\ln{m})/k^{1/2}$. At this point, either event $\event(z_i)$ does not happen, in which case, by~\Cref{clm:case2}, we again obtain
		condition $(i)$ of the lemma; or the event $\event(z_i)$ happens, which by~\Cref{clm:case3} and the choice of $\rw_i$ in~\Cref{eq:rw} implies  
		\[
			\Pr_{X_i}\paren{\dense_f(\Win_m,i) \leq \frac{2}{m} \mid s_{<i},x_{<i}} \leq 4 \cdot \left(\frac{4 \cdot \ln{m}}{k^{1/2}} + \frac{m}{2^{k^{m-i}}}\right) \ll \frac{1}{k^{1/3}},
		\]
		as $i \leq m-1$ and thus $m/2^{k^{m-i}} \leq m/2^{k} \ll 1/k^{1/3}$, as $k=m^{m}$. 
		Taking the union bound over the above two events, we also obtain condition $(ii)$ of the lemma. 
\end{proof}

Finally, we use this lemma to conclude the proof of~\Cref{lem:main}. 

\begin{proof}[Proof of~\Cref{lem:main}]

	Let $T_{1},T_2 \subseteq [m]$ denote, respectively, the iterations in which condition $(i)$ or condition $(ii)$ of~\Cref{lem:iteration} happens. 
	Note that $T_1$ and $T_2$ are random variables over the randomness of $S_i$'s and $X_i$'s. 
	We first claim that with high probability $\card{T_2} < m/2$. This is because 
	for any iteration $i \in T_2$ and any choice of $(s_{<i},x_{<i})$ of prior iterations, by~\Cref{lem:iteration},
	\begin{align*}
		\Pr_{X_i}\paren{\dense_f(\Win_m,i) \leq \frac{2}{m} \mid s_{<i},x_{<i}} \leq \frac{1}{k^{1/3}}.
	\end{align*}
	A union bound on at most $m$ choices for indices on $T_2$ then implies that with probability at least $1 - m/k^{1/3}$, we have 
	$\dense_f(\Win_m,i) > \frac{2}{m}$ for all $i \in T_2$. But then conditioned on this event, the size of $T_2$ cannot be $m/2$ or larger 
	as otherwise $\Win_m$ contains $m/2$ disjoint sets each of which contains than a $2/m$ fraction of the window, which is a contradiction. 
	Thus, 
	\[
		\Pr(\card{T_2} \geq m/2) \leq \frac{m}{k^{1/3}} \ll \frac{1}{k^{1/4}} \tag{as $k=m^{m}$}. 
	\]

	We condition on the complement of this event in the following, namely, that $\card{T_2} < m/2$. Let $\set{i_1,\ldots,i_{m/2}}$ denote the first $m/2$ indices of $T_1$ which by the conditioning on the size of $T_2$ is well defined. 
	We have, 
	\begin{align*}
		\Pr(\text{for all $j \in [m/2]$:}~f(X_{i_j}) = i_j) &= \prod_{j \in [m/2]} \Pr\paren{f(X_{i_j}) = {i_j} \mid f(X_{i_1}) = i_1, \ldots, f(X_{i_{j-1}}) = i_{j-1}} \\
		&\leq \left(\frac{101}{m}\right)^{m/2} \tag{since these are type $(i)$ iterations and we can apply condition $(i)$ of~\Cref{lem:iteration}}.
	\end{align*}
	Putting these two together, combined with the value of $k = m^{m}$, implies that, 
	\[
		\Pr_{(X_1,\ldots,X_m)}\paren{\text{$\forall i \in [m]:$  $f(X_i) = i$}} \leq \frac{1}{k^{1/4}}+ \left(\frac{101}{m}\right)^{m/2} \leq m^{-\eta \cdot m}, 
	\]
	for some constant $\eta > 0$ (taking $\eta = 1/100$ certainly suffices). This concludes the proof. 
\end{proof}

\clearpage

\subsection*{Acknowledgements} 

William Kuszmaul was partially sponsored by the United States Air Force Research Laboratory and the United States Air Force Artificial Intelligence Accelerator under Cooperative Agreement Number FA8750-19-2-1000. The views and conclusions contained in this document are those of the authors and should not be interpreted as representing the official policies, either expressed or implied, of the United States Air Force or the U.S. Government. The U.S. Government is authorized to reproduce and distribute reprints for Government purposes notwithstanding any copyright notation herein.

Additionally, Sepehr Assadi was supported in part by a
NSF CAREER grant CCF-2047061, a Google Research gift, and a Fulcrum award from Rutgers Research Council. Martin Farach-Colton was supported by NSF grants CCF-2118620 and CCF-2106999.
William Kuszmaul was supported by an NSF GRFP fellowship and a Hertz fellowship.

\bibliographystyle{alpha}
\bibliography{./new,./bib}

\clearpage

\appendix
\part*{Appendix}

\section{Proofs of Standard Results in Fractional Coloring}\label{app:fc} 

We prove~\Cref{prop:chif-product,prop:chif-distribution} here for completeness. These proofs are standard; see, e.g.~\cite{scheinerman2011fractional}. We start by presenting the dual view of fractional colorings that is the key to these proofs.

\paragraph{The dual view of fractional colorings.} Given that $\chi_f(G)$ is defined as a solution to an LP, we can use duality to also express $\chi_f(G)$ via the following LP: 

\begin{align}\label{eq:chif-dual}
	\chi_f(G) := 
	\max_{y \in \IR^{V(G)}_{\geq 0}} &~~ \sum_{v \in G} y_v \quad \text{subject to} \quad \sum_{v \in I} y_v \leq 1 \quad \forall I \in \IS{G}. 
\end{align}

This LP is a fractional relaxation of the \emph{clique number} of $G$, namely, the size of the largest clique in $G$ (since, in any integral solution to this LP, the $y$-values that are $1$ must be on the  vertices of a clique).  Interestingly, although the chromatic number and clique size are not duals, their relaxations are.

\begin{proposition*}[Restatement of \Cref{prop:chif-product}]
	Let $G_1 = (V_1,E_1)$ and $G_2 = (V_2,E_2)$ be arbitrary graphs. Define $G_1 \vee G_2$ as a graph on vertices $V_1 \times V_2$ and define an edge between vertices $(v_1,v_2)$ and $(w_1,w_2)$ whenever $(v_1,w_1)$ is an edge in $G_1$ 
	\underline{or} $(v_2,w_2)$ is an edge in $G_2$. Then, $\chi_f(G_1 \vee G_2) = \chi_f(G_1) \cdot \chi_f(G_2)$. 
\end{proposition*}
\begin{proof}[Proof of \Cref{prop:chif-product}]
	We first prove that 
	\begin{align}
	\chi_f(G_1 \vee G_2) \geq \chi_f(G_1) \cdot \chi_f(G_2). \label{eq:chif-product1}
	\end{align}
	Let $y^1 \in \IR^{V_1}$ and $y^2 \in \IR^{V_2}$ be optimal solutions to the dual LP given by~\Cref{eq:chif-dual} for $G_1$ and $G_2$, respectively. 
	Consider the assignment $y \in \IR^{V_1 \times V_2}$ where $y_{u_1,u_2} = y^1_{u_1} \cdot y^2_{u_2}$. We clearly have that 
	\[
	\sum_{(u_1,u_2) \in V_1 \times V_2} y_{u_1,u_2} = \paren{\sum_{u_1 \in V_1} y^1_{u_1}} \cdot \paren{\sum_{u_2 \in V_2} y^2_{u_2}} = \chi_f(G_1) \cdot \chi_f(G_2).
	\]
	We now argue that $y$ is also a valid solution to the dual LP given by~\Cref{eq:chif-dual} for $G_1 \vee G_2$. Fix any independent set $I \in \IS{G_1 \wedge G_2}$. By the definition of the product, we know that $I$ can be written
	as a product set, namely, $I=I_1 \times I_2$ for $I_1 \in \IS{G_1}$ and $I_2 \in \IS{G_2}$. Thus, 
	\[
		\sum_{(u_1,u_2) \in I} y_{u_1,u_2} = \paren{\sum_{u_1 \in I_1} y^1_{u_1}} \cdot \paren{\sum_{u_2 \in I_2} y^2_{u_2}} \leq 1 \cdot 1 = 1,
	\]
	where the inequality is by the constraint of dual LP for $y^1$ and $y^2$ each. Thus, $y$ is a solution to the dual LP for $G_1 \vee G_2$, proving~\Cref{eq:chif-product1}.  
	
	We now prove that 
	\begin{align}
	\chi_f(G_1 \vee G_2) \leq \chi_f(G_1) \cdot \chi_f(G_2), \label{eq:chif-product2}
	\end{align}
	using the primal LP instead. Let $x^1 \in \IR^{\IS{G_1}}$ and $x^2 \in \IR^{\IS{G_2}}$ be optimal solutions to primal LP from~\Cref{eq:chif} for $G_1$ and $G_2$, respectively. 
	Consider the assignment $x \in \IR^{\IS{G_1 \vee G_2}}$ where $x_{I} = x^1_{I_1} \cdot x^2_{I_2}$, using the fact from the previous part that $I = I_1 \times I_2$ for $I_1 \in \IS{G_1}$ and $I_2 \in \IS{G_2}$.  
	
	We again clearly have that 
	\[
	\sum_{(u_1,u_2) \in \IS{G_1 \vee G_2}} \hspace{-20pt} x_{I} = \paren{\sum_{I_1 \in \IS{G_1}} x^1_{I_1}} \cdot \paren{\sum_{I_2 \in \IS{G_2}} x^2_{I_2}} = \chi_f(G_1) \cdot \chi_f(G_2),
	\]
	so it remains to prove that $x$ is a valid solution to the primal LP from~\Cref{eq:chif} for $G_1 \vee G_2$. Fix any vertex $(u_1,u_2) \in V_1 \times V_2$ and consider all independent sets $I_1 \subseteq V_1$ that contain $u_1$
	and $I_2 \subseteq V_2$ that contain $u_2$. Then, $I_1 \times I_2$ is also an independent set in $G_1 \vee G_2$ that contains $(u_1,u_2)$. 
	 Thus,  
	\[
		\sum_{I \in \IS{G_1 \vee G_2}: (u_1,u_2) \in I} \hspace{-20pt} x_{I} \geq \paren{\sum_{I_1 \in \IS{G_1,u_1}}\hspace{-5pt} x^1_{I_1}} \cdot \paren{\sum_{I_2 \in \IS{G_2,u_2}} \hspace{-5pt} x^2_{I_2}} \geq 1 \cdot 1 = 1,
	\]
	where the inequality is by the constraint of primal LP from~\Cref{eq:chif} for $x^1$ and $x^2$ each. Thus, $x$ is a solution to the primal LP from~\Cref{eq:chif} for $G_1 \vee G_2$, proving~\Cref{eq:chif-product2}. 
\end{proof}

\begin{proposition*}[Restatement of \Cref{prop:chif-distribution}]
	For any graph $G=(V,E)$, 
	\[
	    \chi_f(G) = \max_{\textnormal{distribution $\mu$ on $V$}} ~ \min_{I \in \IS{G}} ~ \Paren{\Pr_{v \sim \mu} \paren{v \in I}}^{-1}.
	\]
\end{proposition*}	
\begin{proof}[Proof of~\Cref{prop:chif-distribution}]
Let $\mu$ be any distribution on $V(G)$ and 
define $b:= \max_{I \in \IS{G}} \Pr(v \in I)^{-1}$. 
Create $y \in \IR^{V(m)}$ such that $y_v = b \cdot \mu(v)$ for every vertex $v \in V(m)$ where $\mu(v)$ is the probability of vertex $v$ under the distribution $\mu$. We claim that $y$ is a feasible dual solution in~\Cref{eq:chif-dual}. 

For every independent set $I \in \IS{G}$,  
\[
	\sum_{v \in I} y_v = b \cdot \sum_{v \in I} \mu(v) = b \cdot \Pr_{v \sim \mu}\paren{v \in I} \leq 1,
\]
by the definition of $b$. Thus $y$ is a feasible dual solution. Moreover, 
\[
    \sum_{v \in V(G)} y_v = b \cdot \sum_{v \in V(G)} \mu(v) = b.  
\]
As the dual LP in~\Cref{eq:chif-dual} is a maximization LP, we have that $\chi_f(G) \geq b = \max_{I \in \IS{G}} \Pr_{\mu}(v \in I)^{-1}$, for any distribution $\mu$ on the vertices.  

Conversely, let $y$ be any optimal solution to the dual LP 
and let $c := \sum_{v \in V} y_v$. Define a distribution $\mu$ on the vertices $V$ by setting $\mu(v) = y_v / c$. For any 
independent set $I \in \IS{G}$, we have, 
\begin{align*}
    \Pr_{v \in \mu}\paren{v \in I} = \sum_{v \in I}\mu(v) = \sum_{v \in I} y_v / c \leq 1/c, 
\end{align*}
where the final inequality is because $y$ is a feasible dual solution. Thus, there exists a distribution $\mu$ such that
$\chi_f(G) = c \leq \max_{I \in \IS{G}} \Pr_{\mu}(v \in I)^{-1}$. 

Combining these two parts concludes the proof. 
\end{proof}

\section{Covering The Full Range of the Universe Size}\label{sec:allu}
We now generalize the proof of~\Cref{thm:main} to the full parameter range specified in the theorem. Consider $u$ and $n$ satisfying $$n2^{2^{\sqrt{\log\log n}}} \le u \le 2^{n^{n^2 + n}}.$$
Notice that, on the lower-bound side, we are actually covering a slightly larger range (and therefore proving a slightly stronger result) than required to establish \Cref{thm:main}.

Set
\[
m = (\log \log u)^{1/6}
\quad \text{and} \quad 
k = n / m = n / (\log \log u)^{1/6}.
\]
Note that the setting of $k$ implicitly requires that $(\log \log u)^{1/6} \le n$, which follows from the fact that $(\log \log u)^{1/6} \le (n^2 + n)^{1/6} \le \sqrt{n}$.

The $k$-fold conflict graph $G^{\oplus k}(m)$ has $\log{\chi_f(G^{\oplus k}(m))} = \Omega(n \log m) = \Omega(n \log \log \log u)$ as already argued in~\Cref{sec:small-universe}. To complete the proof, we must establish that the graph $G^{\oplus k}(m)$ has vertices that are subsets of a universe whose size $u'$ satisfies $u' \le u$. Solving for $u'$, we have that
\begin{align*}
    u' & = k M 
        = \frac{n}{(\log \log u)^{1/6}} \cdot 2^{m^{m^2 + m}} 
        \le n \cdot 2^{2^{m^3 / 2}} 
        \le n \cdot 2^{2^{\sqrt{\log \log u}/2}}.
\end{align*}
%Using the fact that $u \ge n \cdot 2^{2^{\sqrt{\log\log n}}}$, we have also that
On the other hand, $u \ge n \cdot 2^{2^{\sqrt{\log\log u}}}$. 
It follows that
\begin{align*}
    \frac{u}{u'} \ge \frac{2^{2^{\sqrt{\log\log u}}}}{2^{2^{\sqrt{\log \log u}/2}}} \gg 1,
\end{align*}
which completes the proof of~\Cref{thm:main} for any choice of $u$ between $n \cdot 2^{2^{\sqrt{\log\log{n}}}}$ and $2^{n^{n^2+n}}$. 

\medskip

Finally, we remark that the term $2^{n^{n^2+n}}$ in the 
upper bound is not tight and can be replaced by any other $2^{2^{\text{poly}(n)}}$ term; this is simply because for any $u = 2^{2^{\text{poly}(n)}}$, $\log\log\log{u} = \Theta(\log{n})$
and thus for any larger universe size $u$ also, we can simply focus on the smallest $2^{n^{n^2+n}}$ numbers in the universe and still obtain the same asymptotic lower bound. The lower bound term is also not tight and can be replaced
with $n \cdot 2^{2^{(\log\log{n})^{\eps}}}$ for any constant $\eps \in (0,1/2)$ by the same argument.

\end{document}